\let\clineorig\cline  

\documentclass[acmsmall,nonacm]{acmart}
\usepackage{multirow}
\let\cline\clineorig

\usepackage{graphicx}%
\usepackage{amsmath}
\usepackage{amsthm}%
\usepackage{mathrsfs}%
\usepackage{xcolor}%
\usepackage{textcomp}%
\usepackage{manyfoot}%
\usepackage{booktabs}%
\usepackage{algorithm}%
\usepackage{algorithmicx}%
\usepackage{algpseudocode}%
\usepackage{listings}%
\usepackage{lstlangcoq}
\makeatletter
\AtBeginDocument{%
  \let\c@figure\c@lstlisting
  
  \let\ftype@lstlisting\ftype@figure 
}
\makeatother

\newtheorem{theorem}{Theorem}
\begin{document}

\title{Verification of a Generational Garbage Collector}


\author{Shengyi Wang}
\email{wangshengyi@sqz.ac.cn}
\orcid{0000-0002-2286-8703}
\affiliation{%
  \institution{Princeton University and Shanghai Qi Zhi Institute}
  \city{Shanghai}
  \country{China}
}

\author{Kathrin Stark}
\email{K.Stark@hw.ac.uk}
\orcid{0000-0002-7086-6518}
\affiliation{%
  \institution{Heriot-Watt University}
\city{Edinburgh}
   \country{Scotland}
}

\author{Andrew W. Appel}
\email{appel@princeton.edu}
\orcid{0000-0001-6009-0325}
\affiliation{%
  \institution{Princeton University}
 \city{Princeton}
 \state{NJ}
   \country{USA}
}

\begin{abstract}
We have formally verified in Rocq+VST a multi-generation 
collector with support for mutable references, written in C and
compatible with OCaml data types.  Its carefully specified API
supports C programs or CertiRocq (a verified compiler from Rocq to C)
is similar to that of OCaml's collector.
We have demonstrated the adequacy of 
our API specification for the mutator (client of the garbage
collector) by verifying client programs.
Our program and our verification
are modular so that
(1) the API spec is independent of the implementation
(e.g., the choice of copying vs. mark-and-sweep, generational-vs-nongenerational)
and (2) the specification and verification of components
of the implementation (such as the forwarding function)
are independent of other components (e.g., design decisions
regarding older generations, multiple threads,
or ``remembered sets'' of mutable references).
\end{abstract}




\maketitle

\section{Introduction}\label{sec1}

We have formally verified the correctness of a multi-generation copying garbage collector
with support for mutable references and updatable arrays, compatible with
OCaml data representations.  We have done so in a modular way, so that
some of the component functions \emph{and their verifications} could be repurposed
into other garbage-collector frameworks, such as the OCaml 4 or OCaml 5 garbage 
collectors.  Our Rocq proofs are available open-source (see \autoref{appendix:repo}).

To provide assurance that the formal specification of the collector is strong enough, 
we have furthermore verified clients of the collector (``mutator" programs);
and we have designed that specification to abstract
(hide from the client) unneeded information, such as the fact that the collector
is generational, or the details of its remembered set for mutable references.
Our interface (and its specification) can be used from the C source-language,
and reasoned about in the C semantics; it does not require reasoning at
the assembly-language level and does not require any special support from
the C compiler.  Thus it also supports proved-correct compilers from functional languages
to C, and we have demonstrated this \cite{korkut25:popl} with the CertiRocq verified compiler for Rocq \cite{paraskevopoulou21:certicoq}.

\paragraph{Background}
Programming languages that allow dynamic allocation of pointer data structures 
generally have a mechanism for reclaiming objects or records that are no longer needed.
Some languages, such as C, leave this task entirely to the programmer---which
makes APIs more complex, makes programming more difficult, and leads to bugs
and unsafety.  Other languages, such as Rust, have a type system that keeps track
of single versus shared ownership (multiple pointers to the same object) and allows
type-directed (and safe) deallocation; this can be highly efficient, but
there is still some programming and API burden (and restrictions on sharing patterns).
Finally, languages such as Lisp, Scheme, ML, Haskell, Java, or Python have
automatic garbage collection (g.c.): graph-traversal algorithms determine which
objects are no longer in use and reclaim them automatically.
Programmers in such languages don't need to keep track of shared versus single-owner
data structures, but pay a modest price both in the overhead of object descriptors
(so the g.c.\ can understand how to traverse the graph) and in the cost of 
executing the graph traversals.

Garbage collection has a long history with thousands of papers published since 1960
\cite{jones23}.  Some of the principles are,
\begin{itemize}
    \item \emph{Live} objects are those reachable from \emph{roots} which are the global and/or local variables of the program.
    \item The location of these roots must be described to the collector, with particular
    attention to which local variables are still in use.\footnote{A local variable is \emph{live} in the terminology of compiler optimization if its contents might be fetched in the future.
    This determination is Turing-complete but can be conservatively approximated by optimizing
    compilers.
    A g.c.-heap object is \emph{reachable} if there is a chain of pointers from a live
    root to the object.  One could, in principle, say that a g.c.-heap object is ``live"
    only if the program will actually traverse such a path in the future, but this
    determination is also Turing-complete and not easy to approximate, so in practice we say that g.c.-heap objects are \emph{live} if they are \emph{reachable from a live root}.}
    \item \emph{Mark and sweep} collection leaves live objects in place, and reclaims the 
    memory used by garbage objects by putting them into a data structure that
    can be used by the allocator.
    \item \emph{Copying} collection moves live objects out of the space they were in, the \emph{from-space},
    into a \emph{to-space}, leaving a large contiguous \emph{from-space} as a region of memory available for allocation.
    \item \emph{Generational} collection relies on the observations that (1) new objects, when created and initialized, necessarily point to older objects; and (2) older objects that
    have already survived garbage collection(s) will likely have a longer future lifetime than younger objects.
    Therefore, if mutable update (storing new
    pointer values into existing objects) is rare, then the mostly preserved invariant
    of newer objects pointing to older objects can be exploited to focus g.c.\ attention
    on just the smaller, newer object spaces.
    \item When mutable update is permitted, generational collection typically requires a \emph{remembered set} keeping track of object fields that may house pointers
    from older objects into younger ones.
    \item Generational collectors may be purely mark-and-sweep or purely copying, or may be organized with copying collection from nurseries to an older generation, and mark-and-sweep in the older generation.
    \item \emph{Shared-memory concurrent} programs are typically arranged so that each
    thread has its own nursery (which eliminates synchronization during object
    allocation) but older generations are shared between threads (requiring synchronization
    when starting garbage collections).
\end{itemize}

All of these principles were established decades ago, though there have been
significant new developments in, for example, the details of how shared-memory
concurrent programs should manage private nurseries with shared older generations \cite{ocaml5gc:icfp}.

In this paper we explain a simple multi-generation single-thread copying collector, 
we explain its invariants  and the specifications (preconditions/postconditions)
of each of its component functions;
and we discuss how those verified functions could be repurposed in other configurations
of a collector.

A multi-generation collector such as we have implemented is not necessarily
optimal.  ``Using multiple generations has a number of drawbacks"
\citep[\S 9.5]{jones23}.  
The state of the art seems to be hybrid 2-generation collectors
that use copying collection for the nursery and mark-sweep collection
for the older generation \cite{jones23}\cite{ocaml5gc:icfp}.
Therefore we have organized both our collector and its proof so
that many important components can be used in such collectors.

\paragraph{Contributions}
We build on previous work by \citet{wang19:oopsla}, who verified a multi-generation copying collector for a subset of the OCaml data formats.  In the current work,
\begin{itemize}
    \item \textbf{We add functionality:} mutable references and updateable arrays.
    \item \textbf{We add expressiveness:} with lemmas regarding
    the ability to allocate new objects, and in general, by strengthening the
    specification enough so that 
    mutator programs can actually be proved correct.
    \item \textbf{We fix bugs:}  The previous collector was verified correct in Rocq w.r.t.\ its specification.  But when \citet{korkut25:popl} verified clients of this collector, they found that this specification was too weak to prove the clients correct.  When we strengthened the specification, we found that
    the C program (the g.c.) did not satisfy the correct specification.  At least one
    of the bugs was a corner case unlikely to have been uncovered in testing.  We fixed the bugs.
    \item \textbf{We ensure modularity:} e.g., by proving the \lstinline{forward} function correct w.r.t.\ a specification
    usable in other g.c.\ designs.
    \item \textbf{We used AI to adapt the proofs:} The previous work, a proof of a multi-generation collector \emph{without} mutable references, needed significant adaptation of its internal invariants and proofs in Rocq.  We used OpenAI Codex to adjust invariants and adapt the Rocq proofs, as described in a separate technical report \cite{wang26:codex}.  However, the adjustments (in previous bullet points) to the client-side specifications were not done with AI.   
\end{itemize}
This collector was designed and built primarily as a g.c. for CertiRocq \cite{anand2017certicoq}, a verified compiler for Rocq.  Therefore, plug-compatibility with OCaml's
g.c. is not an explicit goal.  Our object formats are the same
as OCaml's (though we omit \emph{lazy/force}
tags and fiber-continuation objects); our mutator interface
is similar but not identical.  CertiRocq generates C code
(compilable with CompCert or with LLVM or gcc, though only the
CompCert path is fully verified); our collector's API is compatible
with both CertiRocq-generated code and handwritten C code for 
low-level libraries \cite{korkut25:popl}.

One might think that since Rocq's Gallina programming language is
purely functional, a Rocq program does not need mutation.  However,
there are several uses for mutation: Rocq's built-in persistent
arrays with a functional API but an imperative implementation \cite{coq-arrays2010}, Rocq's coinduction (if implemented
using lazy evaluation with thunk update), or
in functional programs that access mutation via a monadic interface
\cite{korkut24}.  Furthermore, our collector could also
be used for semifunctional languages such as OCaml or Standard ML
(which have explicit mutation) or for lazy languages such as Haskell
(which need mutation for thunk update).

\paragraph{Other verified collector+mutator combinations}

There are many papers describing pencil-and-paper proofs of g.c. \emph{algorithms} (starting with \citet{dijkstra68});
or mechanized proofs of \emph{abstract models} of g.c. programs but not proofs of the actual programs \cite{yang22:zgc}; 
or proofs of garbage collector programs without a demonstration
that the specification is strong and detailed enough to actually prove
mutators correct \cite{shamsu24}.

However, an important aspect of our work is \emph{a usable specification
interface}, with the ability to practically verify
correctness of mutator programs that are clients of our collector,
and the ability to practically verify other collectors that satisfy the same
specification interface.  By verifying the mutator and collector in the
same program logic, or at least in program logics compatible
with the same operational semantics, one can prove end-to-end
correctness of the whole collector+mutator system.

A few others have also verified across that interface, but in each
case there are significant limitations on both the API and on
the tricky combination of mutable references and generational collection:
\begin{itemize}
\item \citet{mccreight2007general} verified three nongenerational
collectors: mark-and-sweep, Cheney-style copying, Baker's concurrent
copying.  They also specified an interface to the mutator
and verified the correctness of a mutator (client) program.
The programs are written,
and the interfaces are specified, in the SCAP (Stack-based Certified Assembly Programming) language, which means that it may be difficult to port
to mutators not written in SCAP, or variations of the collector
not written in SCAP.  Baker's algorithm is \emph{concurrent}---so one
must reason about race conditions---and uses a \emph{read barrier}---that is,
dereferencing (loading) an object field requires more than just a simple
load instruction, so the compiler for the mutator must be rather specialized
in how it generates code.
\item
\citet{mccreight2010certified} verifed a nongenerational Cheney-style 
copying collector, with a specification of the mutator interface,
and verified the correctness of client programs.  Programs
are written in CompCert C.  Standard CompCert compiles through a series
of intermediate languages, each using a memory model with abstract
block numbers, and the semantics is insensitive to renumbering of blocks.
\citeauthor{mccreight2010certified}'s collector inserts a new
intermediate language \emph{GCminor} into this sequence, which takes
advantage of this insensitivity to make the garbage collector's
rearrangement of data transparent to the mutator, and to
proofs of correctness of the mutator.  However, (unlike the C language
in which our system is specified) there are no program logics
or semiautomated reasoning systems for GCminor, with which to prove client programs correct.
\item
\citet{sandberg19} verified a two-generation collector for CakeML,
including the verification of its mutator client, that is, ML
programs compiled to the ``StackLang'' intermediate representation.
But their work has some limitations:
They avoid having a general mechanism for
remembered sets, by segregating
mutable references into a region of the heap where they are never collected.
This is an idiosyncratic approach chosen for convenience in verification,
and it will not scale well to applications that allocate 
large data structures containing mutable references that eventually
will need to be reclaimed.  
They use a two-generation copying collector, which will not
perform well at scale compared either to a hybrid collector or
a multi-generation copying collector \citep[\S 9.5]{jones23}.
The specification of their collector-mutator interface is
tied so tightly to a low-level intermediate language of CakeML
that it is difficult to see how to use the collector in other 
contexts---however, there are mechanized correctness-proving
tools for CakeML programs \cite{Gueneau2017}.
\end{itemize}

\section{Interface between collector and mutator}\label{sec:interface}

In the terminology of garbage collection, the client of the collector (the program
doing useful work) is called the ``mutator."    The mutator builds and traverses
a directed graph of objects.  Each object is represented as $n+1$ words in memory,
with a \emph{header} (or equivalently \emph{descriptor}) and $n$ fields,
each of which may be either a pointer to another heap object, an integer, or
an ``outlier" (pointer to a non-heap object).  Graph nodes represented by integers or outliers have out-degree zero.
Graph nodes may have arbitrary in-degree; that is, there may be sharing or even
cycles in the graph.

We specify and prove our g.c. using the \emph{Verifiable C} program logic,
which is a higher-order separation logic mechanized in the 
Verified Software Toolchain \citep{appel14:plcc,DBLP:journals/jar/CaoBGDA18}, which is itself formalized and executed in Rocq (Coq).  Since Verifiable C is a separation logic,
it naturally accommodates descriptions of tree nodes whose child-pointers point
to disjoint data structures.  But graph structures with sharing are not as natural a fit for separation logic.  To accommodate a graph in VST,
we use the CertiGraph library \cite{wang19:oopsla} that represents an abstract
graph, with a separation-logic embedding as simply the iterated separating
conjunction (``big star'') of all the graph nodes.  

\autoref{subsec:API} presents the API.  First we
will explain the API informally, with pictures and diagrams
indicating objects, fields, pointers, and separated regions of the heap.

Our g.c.\ is designed to support an OCaml-like language and uses OCaml object formats. 
OCaml values can be pointers or \emph{unboxed integers}; to allow the garbage collector
to distinguish between those, (word-aligned) pointer values are represented with even numbers,
and unboxed integers are represented with odd numbers.  In the C program
(for the g.c., or for mutators written in C or code-generated into C), we use
the \lstinline{value} typedef:
\begin{lstlisting}[language=C,label=listing-value,caption=value type that can be either pointer or integer]
typedef void *value $\hspace{1em}$ __attribute((aligned(_Alignof(void *))));
\end{lstlisting}
That is, a \lstinline{value} is a pointer-to-void
(which is a C idiom for a pointer to an unknown type);
the \emph{attribute} is semantically transparent to the
C compiler but serves a hint to our program verifier (VST) that the value might hold an integer instead of a pointer (see \autoref{tag-bit-semantics}). 

\begin{figure}[h]
    \begin{minipage}{2.7in}
    \caption{The mutator has local and global variables that are \emph{roots} of
    the graph objects.  Each object has a 1-word header containing its
    \emph{length} (as well as a \emph{tag} field).  The graph may have
    \emph{outliers} that are not garbage collected but are managed
    by other means.   The one-field object at bottom right of the graph of objects (the dashed-line enclosed region) is unreachable, i.e., garbage.}
    \label{fig:gr}
    \end{minipage}~~~~~\begin{minipage}{2.1in}
    \includegraphics[scale=1]{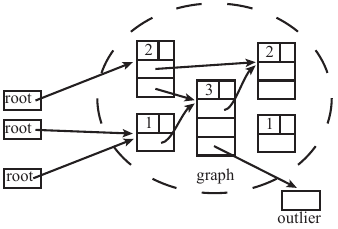}
    \end{minipage}
    \vspace*{-\baselineskip}
\end{figure}

An OCaml object has a 1-word \emph{descriptor}
containing, among other things, the length of the object in words (not counting
descriptor) and an 8-bit \emph{tag} that the mutator may use to distinguish data
constructors.  Each word after the descriptor is a \lstinline{value} unless the tag is $\ge 251$, in which case all
words are considered nonpointers regardless of their last bit.
The \emph{base address} of an object is the address of the first field \emph{after} the descriptor.
\citet[Chapter 23]{minsky22:rwo} provide details.

\begin{figure}[h]
    \hspace*{-.5in}\begin{minipage}{1.9in}
    \caption{The mutator is given a contiguous space for allocation of new
    objects, from address \textsf{alloc} to \textsf{limit}. 
    When an $n$-field object is allocated, the new object will
    be at address \linebreak $\mathsf{alloc}+1\cdot\textsc{word\_size}$;
    \hfill  one must
    \newline
    add $(1+n)\cdot\textsc{word\_size}$ to \textsf{alloc}, and initialize the fields.
    Here, the operation \textsf{root2=C(root1)} is illustrated,
    for a unary data constructor \textsf{C}.
    \newline{}\hspace*{1em}The \texttt{alloc} and \texttt{limit} are components of the
    thread-info; see \autoref{fig:heap}.}
    \label{fig:alloc1}
    \end{minipage}\hspace*{1em}\begin{minipage}{2.5in}
    \includegraphics[scale=1]{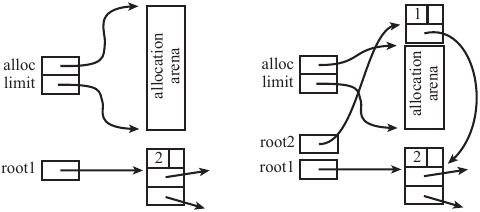}
    \end{minipage}
\end{figure}

\begin{figure}
    \hspace*{-.5in}\begin{minipage}{2.7in}
    \caption{The C program maintains a linked list of frame structs within
    its own function-call stack, using the clever trick \cite{mccreight2010certified} of declaring them as addressable
    local structure variables.  Addressable local structures is one
    of those crazy features of C that can lead, if abused, to 
    dangling pointers, so the mutator's correctness proof must show
    that this is managed well.  The g.c.\ doesn't care
    where this linked list is stored, as long as it is separate
    from the heap and graph.  Ordinary local variables such as $x$ that contain live roots
    must be copied into the frame-struct before calling \lstinline{garbage_collect()},
    or calling functions that might call \lstinline{garbage_collect}.
    The frame-pointer \texttt{fp} is a component of the
    thread-info; see \autoref{fig:heap}.}
    \label{fig:frames1}
    \end{minipage}\hspace*{1em}\begin{minipage}{1.3in}
    \includegraphics[scale=1]{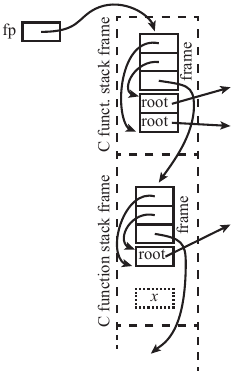}
    \end{minipage}
\end{figure}

The g.c.\ must be able to find the live roots in the mutator's stack and global variables.
Some garbage collectors do this by having the mutator's compiler produce precise
\emph{stack maps} \citep[Chapter 11]{jones23}; but our g.c.\ is meant for mutators
written in C, and C compilers do not produce such stack maps.  Therefore we
adapt a technique described by \citet{mccreight2010certified}---see
\autoref{fig:frames1}.
The mutator maintains a linked list of \emph{frames}, one for each (nominal)
stack frame of the C program and stored as a local addressable struct within
that stack frame.\footnote{We say \emph{nominal} because the C compiler might
choose to inline some function calls, in which case our linked list of root-frames
will not correspond 1-1 to C stack frames but will still be semantically
correct: if the optimizing compiler inlines function $f$ into $g$, the resulting
function's stack frame will have two root-frames in the linked list.}  

The mutator's C program need not keep
the current function's variables in the topmost frame-struct 
while operating upon them; it can use ordinary local variables that
the C compiler will keep in registers.  But 
before each function call, to another C function or to the collector,
the C program must store all local variables \emph{that might be live pointers}
into the topmost frame, and must fetch them back after the call.
The head of this linked list is the \emph{frame pointer} \lstinline[language=C]{fp}.

\begin{figure}
\begin{lstlisting}[language=C]
struct thread_info {
  value *alloc;  $\hspace{4.2em}$ /* beginning of allocation arena */
  value *limit;  $\hspace{4.2em}$ /* end of allocation arena */
  struct heap *heap; $\hspace{1.3em}$  /* heap-management data structure, opaque to mutator */
  struct stack_frame *fp;  /* linked list of root frames */
  unsigned int nalloc; $\hspace{.9em}$   /* number of words needed to allocate after collection */
};
\end{lstlisting}
    \begin{minipage}{3.1in}
    \caption{The mutator maintains a \textsf{thread\_info} structure
    in which it keeps its \textsf{alloc} and \textsf{limit} pointers,
    the \textsf{heap} management pointer,
    its stack-of-frames pointer \textsf{fp},
    and (when calling the g.c.) the number of words \textsf{nalloc} that
    it needs for the next-to-be-allocated object. \newline
    \hspace*{1em}     The \emph{heap management} data structure is an abstract data type.
    That is, the mutator keeps an opaque pointer to it and never traverses it directly.  In that data structure
    are all the free spaces of the older generations (but not the free
    space of the nursery, which is the \emph{allocation arena}),
    \emph{remembered sets} for mutable ref-updates, and the data structures
    to keep track of all the generations.}
    \label{fig:heap}
    \end{minipage}\hspace*{1em}\begin{minipage}{1.7in}
    ~~~\includegraphics[scale=1]{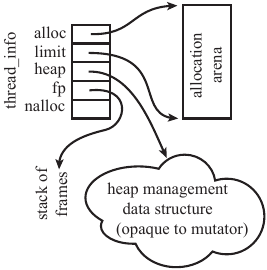}
    \end{minipage}
\end{figure}

The \emph{thread info} structure (\autoref{fig:heap}) allows the mutator and collector to communicate
about available space for allocation.

\begin{figure}
    \begin{minipage}{1.8in}
    \caption{Mutator's view of the g.c.\ interface, at the time of calling \lstinline{garbage_collect}.  The mutator is requesting that the collector should ensure at least 2 words in the allocation arena.
    As illustrated, there appear to be at least two words there already,
    and furthermore the mutator can determine this quickly by
    calculating \lstinline{limit-alloc},
    in which case this call would be unnecessary.\newline 
    \hspace*{1em}The \emph{graph} is quasi-opaque to the mutator's correctness proof: that is, one
    block at a time can be visible using \emph{ramification} lemmas.
    \newline
    \hspace*{1em}The \emph{outliers} region is not garbage-collected but is managed separately
    by the runtime system or foreign functions.}
    \label{fig:call-gc}
    \end{minipage}\hspace*{1em}\begin{minipage}{3.1in}
    \includegraphics[scale=1]{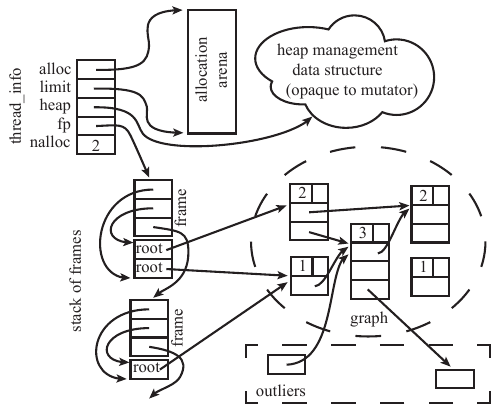}
    \end{minipage}
\end{figure}

\begin{figure}    
    \begin{minipage}{1.7in}
    \caption{Allocating a new record.  This is like \autoref{fig:alloc1} but
    in the larger context.  When not calling \lstinline{garbage_collect}, roots in the topmost frame are typically
    kept in registers (i.e., nonaddressable local variables of the C program),
    as is the thread-info pointer \texttt{ti}.  The mutator stores roots into the topmost frame just
    before calling the collector and fetches them back afterwards
    (because the copying collector may have moved the objects they
    point to).  Here, the mutator has just allocated a new object, by
    adjusting \lstinline{ti->alloc} (old value shown with dashed arrow)
    and making a local variable point to the object (dotted arrow).
    This indicates that two words of memory have left the \emph{allocation space}
    and joined the \emph{graph}.}
    \label{fig:alloc}
    \end{minipage}\hspace*{1em}\begin{minipage}{3.3in}
    \hspace*{-.1in}\includegraphics[scale=1]{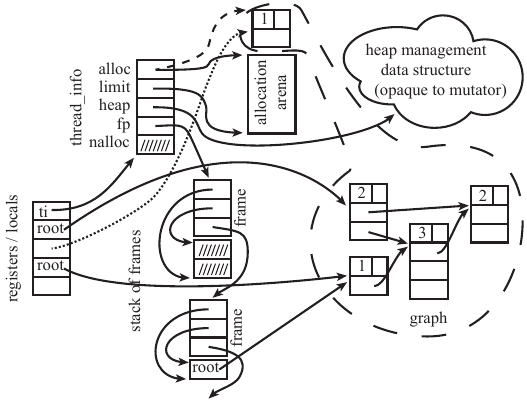}
    \end{minipage}
\end{figure}

All the memory words between \lstinline[language=C]{alloc} and \lstinline[language=C]{limit} are
available for allocation.\footnote{This is similar but not identical to OCaml's interface.}
The \lstinline[language=C]{heap} field points to the g.c.'s internal data structure
that keeps track of memory regions and is opaque to the mutator.
When calling the collector (\autoref{fig:call-gc}), 
the mutator sets the \emph{frame pointer} \lstinline[language=C]{fp} to the
stack-of-frames data structure and sets \lstinline[language=C]{nalloc} to the number of words
of free space needed to allocate the next object(s).

\paragraph{Updating a mutable reference}
A generational garbage collector is permitted to assume that
every pointer is from a younger object to an older object,
\emph{except as explicitly noted when updating a mutable reference
or array}.  The g.c.\ provides an \lstinline{mutable_update} mechanism
by which the mutator can inform the g.c.\ of ref-updates;
see \autoref{fig:update}.  This information is kept in a \emph{remembered set}
within the heap management data structure.

\begin{figure}    

    \begin{minipage}{2.0in}
        \caption{Updating a mutable reference.  The mutator has stored
    a new pointer value (bold arrow) into a field that previously
    pointed elsewhere (dashed arrow).  It must do so by calling
    the \lstinline{mutable_update} function, which (in addition to storing
    the new value as shown) adds the address to the \emph{remembered
    set} (part of the heap management data structure).  One might notice as well that \lstinline{limit} has
    changed (from the dashed arrow to the solid arrow); \lstinline{mutable_update}
    has stolen a word from the end of the allocation space to add
    to the remembered-set data structure.}
    \label{fig:update}
    \end{minipage}\hspace*{1em}\begin{minipage}{3.0in}
    \includegraphics[scale=1]{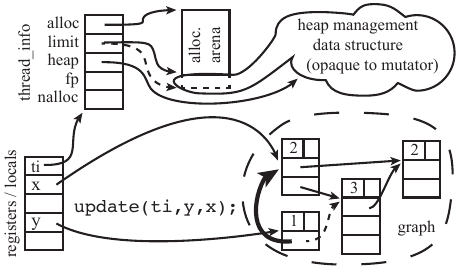}
    \end{minipage}
    \end{figure}

\paragraph{Dereferencing a pointer}
The \emph{graph} (depicted in \autoref{fig:gr} and \autoref{fig:call-gc}) is a quasi-opaque data structure.  That is, the mutator \emph{program} has pointers directly to objects within it, and can directly fetch fields of those objects; but reasoning
about these objects in separation logic is not so simple.
Separation logic is very natural for trees, where we know that
the left subtree and the right subtree occupy disjoint memory regions;
and we can use it to assert that the \lstinline{thread_info} struct
is disjoint from the allocation arena, both of which are disjoint from
the stack of frames, and the heap management data structure, and so on.
But the object graph is not a tree structure; when constructing
$C(x,y)$ for some constructor $C$, it may be that $x$ and $y$ are
the same pointer, or $x$ points within $y$ (there is a path from $y$ to $x$).
To reason about this, we need to \emph{focus} on one object within the
graph as shown in \autoref{fig:deref}.  See, for example, the
\textsc{Localize} rule of \citet[Equation 1]{wang19:oopsla}.

\newcommand{\wand}{\mathrel{-\hspace{-.7ex}*}}
\begin{figure}    
    \begin{minipage}{2.5in}
        \caption{Fetching a field of an object.
        The mutator can fetch field 0 or field 1 of pointer $x$.
        To reason about that in separation logic, the CertiGraph
        library provides \emph{ramification lemmas} for focusing on a single object.
        One such lemma proves that the original graph predicate $g$ 
        is equivalent to 
        $x\mapsto(2,p,q) \ast (x\mapsto(2,p,q) \wand g)$.\newline
        \hspace*{1em}In this illustration, $t$ is the constructor tag.  The mutator can fetch $x[-1]$ and bitwise-and with binary 000...00011111111 to get the
        tag, which it may need for discriminating on data constructors.
        }
    \label{fig:deref}
    \end{minipage}\hspace*{1em}\begin{minipage}{2.5in}
    \hspace*{2em}\includegraphics[scale=1]{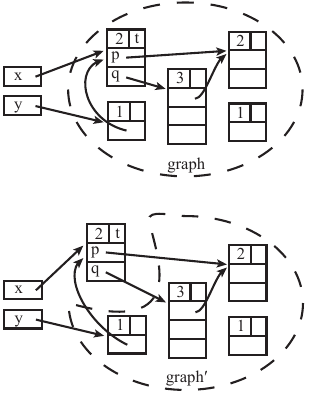}
    \end{minipage}
    \end{figure}

\begin{figure}    
    \begin{minipage}{2.9in}
    \caption{Outlier data structures can point into the graph, but they
    must be registered using \lstinline{mutable_update}.     \newline \hspace*{1em}    
    A function closure (such as the ``2'' object
    in this graph) has a \emph{code pointer} to the start address
    of a machine-language function and an \emph{environment} $e$,
    a value typically represented by graph nodes, representing
    the free variables of the function.  The code pointer is 
    typically not a graph node itself, but needs to be accounted
    for in the separation-logic predicates; we call it an
    \emph{outlier.}  \newline
    \hspace*{1em} Another use for outliers is for opaque data structures managed by a foreign
    function (i.e., a C function in the runtime system, callable
    from the functional program that uses OCaml data types).
    The functional program cannot dereference such pointers,
    but can pass them as arguments back to the foreign function. 
}
    \label{fig:outliers}
    \end{minipage}\hspace*{1em}\begin{minipage}{2.1in}
    \hspace*{1em}\includegraphics[scale=1]{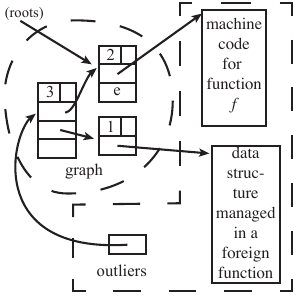}
    \end{minipage}
\end{figure}

\paragraph{Outliers}
Function closures, such as the one in \autoref{fig:outliers},
point into the compiled program, to addresses that may not be
convenient to represent as graph nodes, i.e., addresses outside the heap.\footnote{Unless,
of course, your ML compiler \emph{does} put all machine code into
movable graph objects \cite{appel90:runtime}.}
In the CompCert C semantics, pointer values
may be in an \emph{allocated} state or a
\emph{freed} state.  It is illegal to do any operation
on a freed pointer value, \emph{even address comparison}
\cite{leroy14:mem}.  The garbage collector \emph{does} need
to do address comparisons on function-pointer values,
since any pointer field within a graph object may need
to be tested to see whether it is within the bounds of
the from-space (see \autoref{sec:cheney}).
Therefore, in VST's program logic
(proved sound w.r.t. the CompCert semantics),
to do address comparisons on function pointers,
we need to reason explicitly about the validity
of such pointers.  For that reason, our correctness specification
and proofs maintain an \emph{outlier} set of known-valid pointers.

Another use for outliers is to implement pointers to
objects managed by the runtime system (or by foreign functions).
The mutator cannot directly deference such pointers
(since the focusing-by-ramification lemma applies
only to objects in the graph), but fields of graph objects
can contain such pointers, and the mutator can pass them
to foreign function calls which will know how to traverse
those objects.

Such foreign data structures might even contain pointers
back into the graph.  The foreign function can register
such roots by calling \lstinline{mutable_update}
(see \autoref{subsec:API}).

\subsection{API operations}\label{subsec:API}
In summary, the API for the mutator must support
the following operations:

\begin{description} 
    \item[$\bullet$ \texttt{available(ti,n)}] Test whether there is enough space in the nursery to allocate
    a new object with $n$ fields.
    This is simply \lstinline{(ti->limit)-(ti->alloc) $\ge$ n+1}.
    \item[$\bullet$ \texttt{alloc(ti,$d$,$v_1$,...,$v_{n-1})$}] With precondition that enough space is available, 
    allocate a new object with
    descriptor $d$, with $n$ pointer fields or $n$ words of nonpointer data (depending
    on the descriptor's tag),
    and initialize all the pointer fields with given data.
    This is simple: \lstinline[language=C]{p=ti->alloc+1} is the address of the new object,
    \lstinline[language=C]{ti->alloc+=n+1} allocates the space, 
    \lstinline[language=C]{p[-1]=d} initializes the descriptor word,
    and, for each $i$, \lstinline[language=C]{p[$i$]=$v_i$} initializes the fields.
    \item[$\bullet$ \texttt{dereference(p,i)}]  Given a pointer $p$ to an 
    $n$-field object in the graph and a number $i<n$, fetch the
    $i$th field of the object into a local variable.
    \item[$\bullet$ \texttt{garbage\_collect(ti)}] Invoke the garbage collector to ensure that the allocation space contains least
    \lstinline[language=C]{ti->nalloc}, by reclaiming
    dead objects and/or asking the operating system for more virtual memory.
    \item[$\bullet$ \texttt{mutable\_update(ti,p,x)}] Store value $x$ into
    a field  at address $p$ within an existing object $b$.
    This not only performs \lstinline[language=C]{*p=x;} but also
    (if $x$ is a pointer)
    adds address $p$ to the \emph{remembered set} (see \autoref{fig:update}).  The field offset $(p-b)$ need
    not be a constant; the mutator could compute it, i.e.,
    subscript an array.  Thus, this operation handles 
    \emph{updateable arrays} as well as \emph{mutable references}.
    The function does not need to know the address $b$, only
    the field address.
\end{description}

During all these operations, the mutator has access to the general thread information structure, 
but the heap management data structures and the graph's internal representation remain opaque to it.

\section{Formal specification of the collector}\label{sec:spec}

The mutator builds and traverses a directed graph of object pointers.
The garbage collector's job is to delete unreachable objects and 
rearrange the reachable ones so that the free space is contiguous 
(to allow fast allocation).  The mutator is not supposed to care about
the exact addresses of objects, only the structure of the graph.
Therefore, if the collector removes unreachable objects from the graph,
and moves reachable objects to new addresses while leaving a graph
isomorphic to the reachable subgraph of the original graph,
the mutator should not notice the difference.

We have told this story in pictures.  Now we 
will express those pictures as resource predicates in separation logic
for the different API operations.
Each specification has a precondition \lstinline{PRE} with
a \lstinline{PROP}ositional component,
a \lstinline{PARAM}eter-value component,
a \lstinline{GLOBALs} linkage,
and a spatial component that is
a \lstinline{SEP}arating conjunction of resources.
The \lstinline{POST}condition has
a \lstinline{PROP}ositional component,
a \lstinline{RETURN}-value specifier,
and a spatial \lstinline{(SEP)} component.

\subsection{Garbage collection}
\autoref{list:gcspec} gives an abridged and partially expanded
presentation of the VST \emph{funspec}, or function specification, for
the \lstinline{garbage_collect} function.  In the mechanized
development, the spatial resources shown explicitly here are packaged
in the auxiliary predicate \lstinline{gc_sep}.  We expand them in this
listing so that their roles can be explained individually.

\begin{lstlisting}[language=Coq,numbers=left,numberstyle=\tiny,float,caption={Function spec. of garbage\_collect},label=list:gcspec]
DECLARE _garbage_collect
WITH rsh: share, sh: share, gv: globals, ti: val, g: LGraph, h: heap, 
       alloc: val, avail: Z, heap_p: val, frames: list frame, nalloc: Ptrofs.int,
       roots : roots_t, outlier: outlier_t
PRE [tptr thread_info_type]     $\hspace{2em}$  (* struct thread_info *ti *)
   PROP (readable_share rsh; writable_share sh;
         full_gc g h (frames2rootpairs frames) roots outlier)
   PARAMS (ti) GLOBALS (gv)
   SEP (data_at sh thread_info_type
          (alloc,  (alloc + WORD_SIZE * avail, 
                     (heap_p, (ti_fp frames, Vptrofs nalloc))))
          ti;
        allocation_space alloc avail;
        graph_rep g;
        frames_rep sh frames;
        heap_management h heap_p;
        outlier_rep outlier;
        mem_mgr gv; all_string_constants rsh gv)
POST [tvoid]     $\hspace{4em}$ (* this is a void-returning function *)
   EX g': LGraph, EX frames': list frame, EX h': heap, EX roots': roots_t,
   EX alloc': val, EX avail': Z,
   PROP (full_gc g' h' (frames2rootpairs frames') roots' outlier;
          garbage_collect_relation roots roots' g g';
          frame_shells_eq frames frames';
          Ptrofs.unsigned nalloc <= avail')
   RETURN ()
   SEP (data_at sh thread_info_type
          (alloc',  (alloc' + WORD_SIZE * avail', 
                       (heap_p, (ti_fp frames', Vundef))))
          ti;
        allocation_space alloc' avail';
        graph_rep g';
        frames_rep sh frames';
        heap_management h' heap_p;
        outlier_rep outlier;
        mem_mgr gv; all_string_constants rsh gv).
\end{lstlisting}

\paragraph{The SEP component of the garbage\_collect precondition}~\\
\begin{minipage}{4.2in}
The precondition's first spatial conjunct (\autoref{list:gcspec} line 9) is
a separation-logic ``points-to,'' written in VST as
\lstinline{data_at sh thread_info_type $v$ ti},
where in this case $v$ is a 5-tuple because
\lstinline[language=C]{struct thread_info} has 5 fields.
This \lstinline{data_at} describes the picture
at right 
(compare with upper left of \autoref{fig:call-gc}).
Value \lstinline{ti} is a pointer to a \lstinline[language=C]{struct thread_info} 
with those five field values.
\end{minipage}
\quad
\begin{minipage}{1.3in}
\begin{tabular}{ r | l | }
\cline{2-2} 
ti~$\longrightarrow$ & alloc \\ \cline{2-2}
& limit \\  \cline{2-2}
& heap\_p \\  \cline{2-2}
& fp \\  \cline{2-2}
& nalloc \\ \cline{2-2}
\end{tabular}
\end{minipage}\linebreak
\noindent  The value 
\lstinline{alloc} is the address
of the \emph{allocation space} shown in \autoref{fig:call-gc},
\lstinline{alloc+WORD_SIZE*avail} is the \emph{limit},
\lstinline{heap_p} is the address of the heap-management
data structure, and so on.

The values \lstinline{alloc}, \lstinline{avail}, etc., are 
quantified in lines 2--4; that is, whenever the
\lstinline{garbage_collect} function is called,
(the proof of) the caller must exhibit some values of those quantified variables that
will make all the precondition predicates satisfiable.
In lines 2-4, the \lstinline{val} type is a CompCert value
(typically a pointer value), \lstinline{Z} is Rocq's integer type,
\lstinline{Ptrofs.int} is the type of integers modulo $2^k$
(if integers are $k$ bits and $k=8\cdot \textsc{word\_size}$).

The conjunct \lstinline{alloc_space alloc avail} describes another
part of \autoref{fig:call-gc}, the region in which the mutator may allocate new records:

\begin{tabular}{ r | l | }
\multicolumn{2}{}~\\
\cline{2-2} 
alloc~$\longrightarrow$ & \quad\quad \\
~ & ~ \\
~ & ~ \\
~ & ~ \\ \cline{2-2}
\multicolumn{1}{r}{$\mathrm{alloc}+\mathrm{WORD\_SIZE}\cdot{}\mathrm{avail} ~ \longrightarrow$ } & 
 \multicolumn{1}{l}{} \\
\end{tabular}

\vspace\baselineskip
The conjunct \lstinline{frames_rep sh frames} (fig. ~\ref{list:gcspec} line 15) describes the
linked-list-of-frames data structure, with \lstinline{ti_fp(frames)}
being the start address of the list.

On line 18, \lstinline{heap_management h heap_p} is
the heap management data structure, rooted at address \lstinline{heap_p},
whose contents are described by the \lstinline{h} value.
This conjunct does \emph{not} describe the actual objects of the heap,
just the management data structures.  The structure of the
\lstinline{heap} type need not be known by the mutator
(or by its proof of correctness).

The \lstinline{graph_rep} conjunct (line 14) describes all the
\emph{graph} objects shown in \autoref{fig:call-gc}.
Unlike the other conjuncts, it does not have a single
root address (such as \lstinline{ti}, \lstinline{alloc}, or
\lstinline{ti_fp(frames)}).  Instead, the
labeled-graph type (\lstinline{LGraph})
describes an address for each vertex of graph \lstinline{g}.
To characterize this formally, we describe a directed graph structure in
Rocq using the CertiGraph library \cite{wang19:oopsla} (github.com/CertiGraph/CertiGraph).
This library can describe the embedding into VST separation logic of vertices, edges, vertex labels, edge labels, and so on.  

Objects outside the garbage-collected heap, that may point into the heap
or may be targets of pointers within the heap, we call \emph{outliers}.
That's a separate region of memory characterized by the predicate
\lstinline{outlier_rep}.

The last two \textsc{sep}arating conjuncts (line 18) give the g.c.\ access
to the malloc/free memory manager (in case it wants to create a new
generation, for example) and the read-only string literals
of the program (in case it wants to print an error message).

The current mechanized definition of \lstinline{gc_sep} additionally
contains \lstinline{heap_remset_rep} and \lstinline{remset_rep}, which
represent the collector's concrete and logical remembered-set state.
These internal bookkeeping resources are omitted from
\autoref{list:gcspec}; their role in the correctness argument is
discussed in
\autoref{subsec:remembered-sets-recorded-back-pointers}.

\paragraph{The PROP component of the garbage\_collect precondition}
The \textsc{prop}ositional component of the precondition states that the quantified
values of the WITH clause must satisfy these propositions:
the graph \lstinline{g} must be compatible with the heap \lstinline{h},
the stack of \lstinline{frames}, the \lstinline{roots} within
those frames, and the \lstinline{outliers}.
In addition, the permission-share \lstinline{rsh} (used for
accessing string literals) gives at least \emph{read} permission,
and the share \lstinline{sh} gives the \lstinline{frames}
data structure at least \emph{write} permission.

\paragraph{The garbage-collect postcondition}

The garbage collector produces a new graph \lstinline{g'}
with new \lstinline{roots'} that live within an
adjusted \lstinline{frames'}.
It provides a fresh allocation space beginning at address
\lstinline{alloc'} containing \lstinline{avail'} words;
and it updates its heap-management data structure to abstract
value \lstinline{h'}.
The postcondition starts with existential quantifiers for
all of these values; in the \lstinline{PROP} clause it
asserts that all these values are
compatible with each other;
and in the \lstinline{SEP} clause it asserts that
all these values are represented in memory.

In addition, there are two \lstinline{PROP} clauses
relating the old graph to the new graph (\lstinline{garbage_collect_relation})
and relating the old stack of frames to the new one (\lstinline{frame_shells_eq});
see also \autoref{fig:frame-shells}.  The last \lstinline{PROP}
clause guarantees that the allocation space is big enough
to satisfy the mutator's request.

\begin{figure}    
    \begin{minipage}{2.5in}
        \caption{The copying garbage collector may move any or all of the
        root objects $x,y,z$ to new addresses $x',y',z'$.  But the addresses
        of all the frames $(f_1,f_2,r_1,r_2,r_3)$ in the stack-of-frames must stay the same.
        This relation is expressed by the predicate
        \lstinline{frame_shells_eq}. \newline
        ~\newline
        The function \lstinline{frames2rootpairs} takes an abstract
        \lstinline{list(frame)} value such as this, and extracts
        to the list of address/value pairs $[(r_1,x);(r_2,y);(r_3,z)]$.       
        }
    \label{fig:frame-shells}
    \end{minipage}~~\begin{minipage}{2in}
    \hspace*{2em}\includegraphics[scale=1]{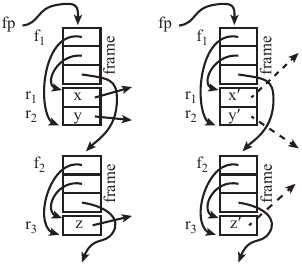}
    \end{minipage}
    \end{figure}

We have proved several properties of \lstinline{garbage_collect_relation $r_1$ $r_2$ $g_1$ $g_2$}.
In particular,  under appropriate conditions (that are satisfied
immediately after a completed collection), there is a graph isomorphism
between the reachable subset of graph $g_1$ with roots $r_1$
and the reachable subset of $g_2$ with roots $r_2$ (\lstinline{garbage_collect_isomorphism}).\footnote{You might think that, after a collection, graph $g_2$ is the same as its reachable subset;
but that is not true of a generational collector, which is not 
obligated to collect all the garbage at every collection.}

\subsection{Proving correctness of mutator operations}\label{subsec:proving-mutator}

\paragraph{Correctness of allocating an object}

\begin{lstlisting}[language=Coq,numbers=left,numberstyle=\tiny,float,caption={Schema of the function specification of \texttt{alloc}. The predicate
\texttt{in\_graph} is typically type-directed and therefore the
values $v_1,\ldots,v_n$ have types $X_1,\ldots,X_n$, respectively.
The notation \texttt{app d [v1, ..., vn]} is an abuse of notation: it
represents the application of a typed constructor, described by
the descriptor \texttt{d}, to a list of arguments. See section~4 of
\citep{korkut25:popl} for further details.},label=list:allocspec]
DECLARE _alloc_spec
    WITH rsh : share, sh: share, gv : globals, ti : val, g : LGraph, 
         roots : roots_t, sh : share, outlier : outlier_t, t_info : thread_info
         p1 ... pn : rep_type, d : cRep, v1 ... vn : X1... Xn
    PRE  [tptr thread_info_type; val; int_or_ptr_type; ...; int_or_ptr_type)
       PROP (readable_share rsh; writable_share sh; 
             full_gc g c t_info roots outlier;
             in_graph g outlier v1 p1, ..., in_graph g outlier vn pn;
             add_node_compatible g [v1...vn];
             n < headroom t_info; 
             mem_mgr gv; all_string_constants rsh gv
             )
       (PARAMS (ti, rep_type_val p1, ..., rep_type_val pn))
       (GLOBALS [gv])
       (SEP (gc_sep g t_info roots outlier ti sh gv))
    POST [ int_or_ptr_type ]
      EX (p' : rep_type) (g' : graph) (t_info' : thread_info),
        PROP (in_graph g' outlier (app d [v1, ..., vn]) p';
              ti_frames t_info = ti_frames t_info';
              gc_graph_iso g roots g' roots;
              headroom t_info' = headroom t_info - Z.of_nat (S n))
        RETURN  (rep_type_val g' p')
        SEP (gc_sep g' t_info' roots outlier ti sh gv).
\end{lstlisting}

Recall that \lstinline{alloc} allocates and initializes a new object with descriptor $d$ and supplied values $v_1,\ldots,v_n$. Operationally, it chooses the address of the new object, advances the allocation pointer by the required number of words, writes the descriptor word, and initializes the object fields. However, a specification that described only this low-level memory update would be too weak to compose with the garbage collector and with later operations on graph-structured values. The actual specification is shown in Figure~\ref{list:allocspec}.

The precondition contains the same assumptions as the garbage-collector specification. In particular, the spatial part describes the thread-info structure, the linked list of stack frames, heap-management data, and the concrete representation of the graph. In the mechanized development, these spatial resources are packaged in
the predicate \lstinline{gc_sep}.  The schematic specification above
suppresses its remembered-set parameters, which are unchanged by this
allocation operation.

There are two additional requirements specific to allocation. First, the allocation area must contain enough free space. 
(We discuss the workflow for checking this condition before calling \lstinline{alloc} in the next subsection.)
For a thread-info structure \lstinline{t_info}, the available space is computed as the distance between the current allocation pointer and the allocation limit; we write this  as \lstinline{headroom t_info} and require this headroom to be larger than $n$.

Second, the specification must describe the values and their location passed as fields of the new object.
Such values may be immediate data, represented by \lstinline{repZ z}; outlier pointers, represented by \lstinline{repOut p}; or pointers to nodes in the graph, represented by \lstinline{repNode v}.
These alternatives are captured by the type \lstinline{rep_type} which can be transformed to the C \lstinline|value| type using the function \lstinline|rep_type_val|.

As the mutator should not know about the exact graph layout in memory, the propositional part of the precondition contains predicates of the form
\lstinline{in_graph g outlier v_i p_i}
that should entail at least how the values are encoded in the graph. 
Usually it is stronger, and in a type-directed manner might add additional propositional information to simplify later reasoning, 
such as the existence of subvalues in the graph. 
The predicate \lstinline{in_graph} is intentionally parameterized, since different clients may use different type-specific representation predicates. Nevertheless, it must satisfy a small of structural properties, as being stable under graph extension and graph isomorphism. 

The postcondition produces a new graph $g'$ and a new thread-info structure \lstinline{t_info'}. Since allocation does not perform garbage collection, the root set and linked list of stack frames are unchanged. Moreover, $g'$ is isomorphic to the old graph $g$ on the previous roots. 

Additionally, the graph contains (at position $p'$) the new object \lstinline|app d [v1,...,vn]| whose descriptor is $d$, with edges from this fresh node to the graph values among $v_1,\ldots,v_n$.
These minimal conditions are fulfilled by the 
``constructor representation predicate``, \lstinline{graph_cRep g p (boxed d n) [p1 ... pn]}
that represents (in this example) 
a boxed constructor (a record pointer,
not an unboxed integer) with tag $d$ and length n,
whose out-edges are labeled with $p1 ... pn$.
To preserve well-formedness of the graph and heap, the newly introduced edges must originate only from the freshly allocated node; their targets must already be valid graph nodes or permitted outliers; and no duplicate edge should be introduced. 
This is ensured by the \lstinline{add_node_compatible} predicate.
These side conditions are satisfied naturally when allocating constructor nodes.

Finally, the postcondition records the effect on the allocation area. The new thread-info structure has the same stack-frame list as before, but its available nursery space has decreased by $n+1$ words: one word for the descriptor and $n$ words for the fields. Although this fact could be recovered by invoking the availability predicate again, storing the updated value of \lstinline{headroom} directly in the postcondition makes subsequent allocation proofs more convenient.

\paragraph{Ensuring enough space for allocation}

To ensure that there is enough space to allocate one or more objects whose size (including descriptors) totals $n$, 
the mutator calls the \lstinline{available} test before allocating (\autoref{list:ensuring}).
\begin{lstlisting}[numbers=left,float,numberstyle=\tiny,caption={Ensuring space to allocate an object},label=list:ensuring,language=C]
if (!available(ti,n)) {
   store root pointers into the topmost frame;
   ti->nalloc=n;
   garbage_collect(ti);
   fetch root pointers back from the topmost frame;
}
\end{lstlisting}

We separate the \lstinline{available} test
from the \lstinline{garbage_collect} call for two reasons. First, it avoids the cost
of storing and fetching root pointers in the common case that no collection is needed.
Second, if an extended basic block (tree of control flow) will predictably allocate
several records whose total size is $\le n$, multiple \lstinline{available} tests can be coalesced
into a single one.

To allocate a new object, the mutator will now always first
ensure
there's enough space.
If it needs to call the garbage collector, it simply passes the \lstinline{ti} pointer.
To \emph{prove} correctness of this call, the mutator proof
must satisfy the precondition of \lstinline{garbage_collect};
that is, demonstrate values
\lstinline{rsh, sh, gv, ti, g, h, alloc,} etc. such that all
the \textsc{prop}ositions in the precondition hold and 
all the \textsc{sep}arating conjuncts (thread-info struct,
allocation space, graph, etc.) exist in memory.

The postcondition of the if-statement in \autoref{list:ensuring}, whether or not the 
\emph{then} branch is taken, is that 
that $\mathsf{limit}-\mathsf{alloc}\ge n$.  
This allows the \lstinline{allocate} operations
illustrated in \autoref{fig:alloc} to be 
proved correct.

\paragraph{Correctness of dereferencing a pointer}
Consider a pointer \lstinline{p} to an object with $n$ fields in the graph and an
index $i<n$. The function \lstinline{dereference(p,i)} retrieves the $i$-th field
of the object into a local variable. Operationally, this corresponds to the C
statement \lstinline{y = x[$i$];}, where \lstinline{x} is the memory address of the
object in the graph. Throughout this section we use the example shown in
\autoref{fig:deref}, with $i=0$.

Since \lstinline{dereference} only reads the graph and never modifies it, its
specification (Listing~\ref{list:dereferencespec}) requires substantially weaker
assumptions about the garbage collector layout than for example allocation. 
The only spatial requirement is that the graph is
represented in memory. Again, the
existence of the particular node being dereferenced are expressed in the
propositional part of the assertion.

Again, the particular node will be described in the propositional part, 
here via the \lstinline|in_graph| predicate. 
In case that (as in \autoref{fig:deref}) $v$ describes a constructor with two arguments, 
recall that this predicate would imply the 
``constructor representation predicate``, \lstinline|graph_cRep g vx (boxed t 2) [p;q]|
that represents (in this example) 
a boxed constructor  with tag $t$ and length 2,
and out-edges that are labeled with $p$ and $q$.
Additionally, this predicate typically implies propositional predicates describing the 
values at $p$ and $q$, added to the postcondition.


In order for the VST separation Hoare logic to dereference a pointer $x$,
there must be a \lstinline{sep} clause of the form
$x\mapsto(v_0,v_1,\ldots)$, written in VST as 
\lstinline{data_at $\pi$ $\tau$ [$v_0;v_1;\ldots$] $x$},
with $\pi$ a permission-share and $\tau$ a C struct or array type.
We employ CertiGraph's \emph{focusing} rule to transform
the precondition into:

\begin{lstlisting}
let $A$:= data_at graph_rep g (Tarray int_or_ptr_type 2) [rep_type_val g p; rep_type_val g q])
in PROP(graph_cRep g vx (boxed t 2) [p;q])
   ...
   SEP($A$; $~$  $A$ -* graph_rep g)
\end{lstlisting}
This is just as illustrated in \autoref{fig:deref}, and the conjunct $A$
is sufficient to fetch \lstinline{x[0]} or \lstinline{x[1]}.

\begin{lstlisting}[float,caption={Funspec for \textsf{dereference}},label=list:dereferencespec]
DECLARE _dereference2 
  WITH g : graph, outlier : outlier_t,
       v: X, vx: rep_type
  PRE [int_or_ptr_type]
  PROP (in_graph g outlier (app v.descr v.proj1 v.proj2) vx)
   PARAMS (rep_type_val g vx)
   GLOBALS (gv)
   SEP (graph_rep g)
  POST [tptr int_or_ptr_type]
  EX  (p q: rep_type) (sh: share),
  PROP (in_graph g outlier x.proj1 p; in_graph g outlier x.proj2 q; writable_share sh)
  RETURN  (rep_type_val g p) 
  SEP (let A := data_at sh (Tarray int_or_ptr_type 2)
                    [rep_type_val g p; rep_type_val g q] (rep_type_val g vx)
    in A * A -* graph_rep g). 
\end{lstlisting}


After the fetch, one can simply use wand adjointness to transform 
\linebreak $A\ast(A \wand \mathsf{graph\_rep ~g})$
back into $\mathsf{graph\_rep ~g}$.

\paragraph{Correctness of updating a mutable reference}

\autoref{list:updatespec}
shows the function specification for 
\lstinline{mutable_update(ti,p,x)}, which sets
a mutable reference (or array slot) at address \lstinline{p} to value \lstinline{x}.
It works on similar principles as those discussed above.

\begin{lstlisting}[float,caption={Funspec for \textsf{mutable\_update}},label=list:updatespec]
Definition int_mutable_update_spec :=
  DECLARE _mutable_update
    WITH ti: val, v: exterior_t, t_info: thread_info, sh: share, g: LGraph,
         it: interior_t, outlier: outlier_t, rh: remset_heap
  PRE [tptr thread_info_type, tptr int_or_ptr_type, int_or_ptr_type]
  PROP (writable_share sh;
        1 <= headroom t_info;
        graph_heap_compatible g (ti_heap t_info).(pt_heap);
        remset_heap_and_heap_compatible rh (ti_heap t_info).(pt_heap);
        mutable_location_compatible g it;
        exterior_compatible g outlier v)
    PARAMS (ti; interior_address it g; exterior2val g v)
    GLOBALS ()
    SEP (graph_rep g;
         outlier_rep outlier;
         before_gc_thread_info_rep sh t_info ti;
         heap_remset_rep g (ti_heap t_info).(pt_heap) rh)
  POST [tvoid]
    EX g': LGraph, EX t_info': thread_info, EX rh': remset_heap,
    PROP (mutable_graph_update g it v g';
          t_info' = decr_info_nursery t_info (exterior2val g v);
          rh' = mtb_upd_remset_heap (exterior2val g v) (RemSetInterior it) rh)
    RETURN ()
    SEP (graph_rep g';
         outlier_rep outlier;
         before_gc_thread_info_rep sh t_info' ti;
         heap_remset_rep g' (ti_heap t_info').(pt_heap) rh').
\end{lstlisting}

\subsection{Adequacy claim}

In \autoref{sec:generational} and \autoref{sec:gc-correct} we will prove that the garbage collector
(with all the operations listed at the beginning of \autoref{subsec:API})
implements its specification (including the funspecs in \autoref{list:gcspec}
and \ref{list:updatespec}, along with supporting lemmas mentioned
in \autoref{subsec:proving-mutator}).  But 
whenever one proves that a program module implements an API specification,
one must wonder whether it is really the right specification.
The most reliable way to demonstrate adequacy of a specification is to use that specification
to prove correctness of clients of that API.
\vspace\baselineskip

\begin{theorem}The g.c. specification is adequate for
proving correctness of mutators that allocate and traverse
data structures.
\end{theorem}
\begin{proof}
Our garbage collector has been used
in the runtime system of CertiRocq, a compiler from Rocq to C that
uses OCaml data representations.  The Verified Foreign Function Interface
(VeriFFI) for CertiRocq \citep[sections 9 and 11]{korkut25:popl} supports proofs of correctness
of client programs that exercise the garbage collector.
\end{proof}

Section 9.2 of \citeauthor{korkut25:popl} summarizes the
Rocq/VST proof of a function \lstinline{uint63_to_nat} that
converts an unboxed 63-bit unsigned integer to a Peano natural number
represented with data constructors S and O (see \autoref{list:uint63-to-nat}).  Lines 5--12 implement the pattern shown in \autoref{list:ensuring}. 
Because the local variable \lstinline{temp} is live after the 
call to \lstinline{garbage_collect}, and it may contain a heap pointer,
it must be stored into the stack of frames (at line 6) and
fetched back afterwards (at line 11).
Lines 7,8,11 implement pushing
and popping from the stack of frames; note that both \lstinline{roots}
and \lstinline{fr} are addressable local variables.
The function \lstinline{alloc_make_nat_S}
is an instance of the \lstinline{alloc} API operation, specialized
to a 1-field record; a precondition of this function is
that at least 2 words of space remain between
\lstinline{alloc} and \lstinline{limit}; this is ensured
by the \lstinline[language=C]{if} statement.

\begin{lstlisting}[language=C,float,numbers=left,numberstyle=\tiny,label=list:uint63-to-nat,caption={A function that allocates; see Figure 2 of \citet{korkut25:popl}}]
value uint63_to_nat (struct thread_info *tinfo, value t) {
  uint64 i = ((uint64)t)>>(uint64)1; /* strip off the tag */
  value temp = make_nat_O(); /* create the base case */
  while (i) {
    if (available(tinfo, 2)) { /* test whether we need to call g.c. */
      value roots[1]={temp}; /* register the root-pointer temp */
      struct stack_frame fr = {roots+1,roots,tinfo->fp};
      tinfo->fp= &fr;
      tinfo->nalloc = 2; /* state the need for 2 words */
      garbage_collect(tinfo);
      temp=roots[0]; tinfo->fp=fr.prev; $~~$ /* pop the frame stack */
    }
    temp = alloc_make_nat_S(tinfo, temp); /* apply S constructor to temp */
    i--;
  }
return temp;
}
\end{lstlisting}

In the function \lstinline{uint63_to_nat},
at the time one decides to allocate a new record (at line 5),
all the values to be used in initializing fields of that record
are already in local variables (in this case, \lstinline{temp}).
This is typical of how programs allocate.
But Section 11 of \citeauthor{korkut25:popl} describes a more
challenging case: converting a list of characters into a packed
bytestring.  In this case, since the length of the string
is dynamic without an \emph{a priori} bound, the values to be
stored cannot all be in local variables at once; during the loop that
initializes the fields, the mutator must also traverse
(dereference) data structures in the graph.  During that loop,
the  \lstinline{full_gc} invariant does not hold, since there is already partially allocated
record.  Therefore, the precondition
(and postcondition) of the \lstinline{dereference} operation
must be (and is) weaker than \lstinline{full_gc}.  

\section{Garbage collector implementation}\label{sec:gcimpl}

We have implemented (and verified) a multi-generation copying collector
inspired by the one that \citet{reppy93:gc} built for Standard ML.
Because we want the specifications and proofs of \emph{components}
of the collector to be independent of the precise details of
how the generations work, we will first describe some abstract
principles.

\begin{lstlisting}[float=b,caption={Rocq code describing the functional model of a space, and a heap as a squence of spaces.},label=list:heap]
Record space : Type := Build_space {
    space_start: val;
    used_space: Z;
    available_space: Z;
    total_space: Z;
    space_sh: share;
    used_leq_available: 0 $\le$ used_space $\le$ available_space;
    available_leq_total: available_space $\le$ total_space;
    space_upper_bound: total_space <= MAX_SPACE_SIZE; }.
    
Record heap : Type := Build_heap {
    spaces : list space;  
    spaces_size : Zlength spaces = MAX_SPACES }.
\end{lstlisting}

\begin{lstlisting}[language=C,float,caption={The C structs characterizing a heap. 
 Not shown are the separation logic predicates \lstinline{space_rep} and \lstinline{heap_rep} that relate these to \lstinline{space} and 
 \lstinline{heap} in the obvious way: \lstinline{used_space} between
 \lstinline{start} and \lstinline{next}, \lstinline{available_space} between \lstinline{start} and \lstinline{limit}, \lstinline{total_space} between \lstinline{start} and \lstinline{rem_limit}.},label=list:struct-heap]
struct space { value *start, *next, *limit, *rem_limit; };
struct heap {  struct space spaces[MAX_SPACES]; };
\end{lstlisting}

\begin{lstlisting}[float,caption={~},label=list:graph-heap-compatible]
Definition graph_heap_compatible (g: LGraph) (h: heap): Prop :=
  let $N_g$ := length g.(glabel).(g_gen)  in 
  (* each active generation is compatible with the graph *)
  Forall (generation_space_compatible g)
         (combine (combine (nat_inc_list $N_g$) g.(glabel).(g_gen)) h.(spaces)) /\
  (* each inactive generation's space-pointer is NULL *)
  Forall (eq nullval) (skipn $N_g$ (map space_start h.(spaces))) /\
  (* There are not more active generations than the length of the array *)
  $N_g$ <= length h.(spaces).
\end{lstlisting}

The \emph{heap} is a sequence of $N_h$ \lstinline{space}s
(see \autoref{list:heap}).  Each space
is a contiguous region of memory
from \lstinline{space_start} to \lstinline{space_start+total_space}.
In each space, the portion up to \lstinline{space_start+used_space}
contains graph objects;
the portion from there to \lstinline{space_start+available_space}
is available for allocating
new objects;
and the remainder up to \lstinline{space_start+total_space}
is used for other purposes.  (The \lstinline{space_sh} is a \emph{permission share},
and other invariants will generally enforce that \lstinline{space_sh} gives
the g.c. at least read+write permission on its heap spaces).

For example, in a simple Cheney-style copying collector there
are two spaces ($N_h=2$).  In each space, \lstinline{available_space} is
equal to \lstinline{total_space}.
The mutator's \lstinline{alloc} is set
to \lstinline{space[0].space_start} and its \lstinline{limit}
is set to \lstinline{space[0].space_start+available_space}.
The mutator allocates in space 0 until it reaches its \lstinline{limit}.
Then the g.c.\ copies (\lstinline{forwards}) all reachable objects into space 1,
and the spaces swap roles (meaning that the mutator's 
\lstinline{alloc} is set to \lstinline{space[1].space_start} before
resuming).

In our collector, there are $N_h$ generations,
each twice as large as the last.

In either of these arrangements, every graph object must live in the
\emph{used} part of one of the spaces.
Each graph vertex is labeled with its space-number
in the range $[0,N_g)$.  Spaces in the range $[N_g,H_h)$ must be marked
as inactive, \lstinline{space_start=nullval}.

The predicate \lstinline{graph_heap_compatible} characterizes these properties
(\autoref{list:graph-heap-compatible}).

\begin{figure}[!tp]
    \includegraphics[scale=0.9]{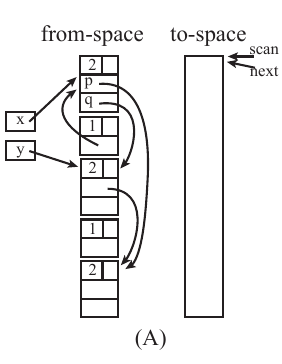}\hspace*{-3em}
    \includegraphics[scale=0.9]{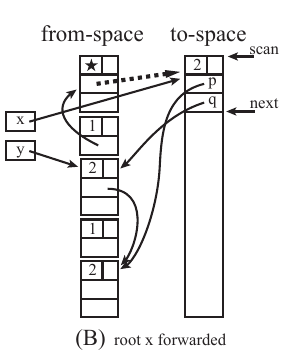}\hspace*{-3em}
    \includegraphics[scale=0.9]{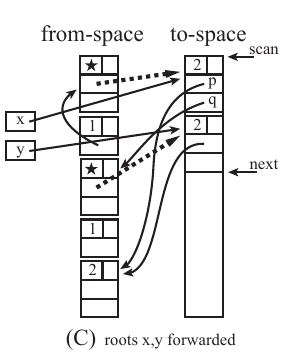}\linebreak
    \includegraphics[scale=0.9]{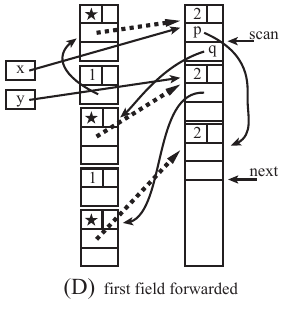}\hspace*{-3em}
    \includegraphics[scale=0.9]{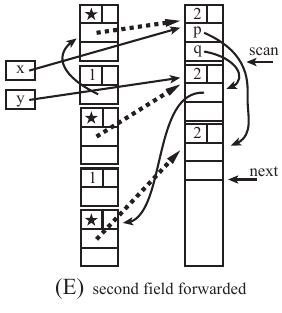}\hspace*{-3em}
    \includegraphics[scale=0.9]{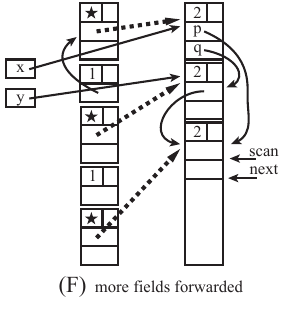}\linebreak
    \includegraphics[scale=0.9]{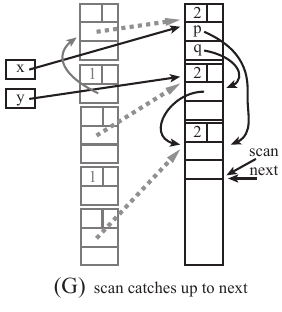}
    \hspace*{4em}
    \begin{minipage}[b]{2.3in}
    \vspace*{\baselineskip}
    \caption{Cheney's algorithm.  At the start, \textbf{scan} and
    \textbf{next} point at the beginning of the to-space.
    First, the roots are forwarded, meaning that the contents of 
    records they point to are copied verbatim and a
    \emph{forwarding pointer} (dotted line) is installed.
    The $\star$ indicates records that have been forwarded.
    Then, each field between \lstinline{scan} and \lstinline{next} is forwarded.
    Eventually \lstinline{scan} catches up with \lstinline{next}, and the algorithm
    is finished; the entire from-space may now be re-used
    for other purposes.}
    \label{fig:cheney}
    \vspace*{\baselineskip}
     \end{minipage}
\end{figure}

\section{Cheney's algorithm}\label{sec:cheney}
At the heart of a generational copying collector is an
algorithm for a copying a \emph{from-space} to a \emph{to-space}---see 
\citet[Chapter 4]{jones23}.  For that we use
Cheney's algorithm \citep[Algorithms 4.2+4.3]{jones23},
with a variation for local depth-first traversal \citep[Section 4.2]{jones23}
to improve cache locality.  \autoref{fig:cheney} illustrates
Cheney's algorithm.

Garbage collectors traverse object graphs, by depth-first search or breadth-first search or by other means.
Depth-first search requires an auxiliary stack; breadth-first
search requires an auxiliary queue.  One can avoid a separate stack by using
\emph{pointer reversal} within the graph, but this has high overhead.  The beauty
of Cheney's breadth-first algorithm is that the space between \lstinline{scan}
and \lstinline{next} \emph{is} the queue, so no auxiliary structure is required.
But the disadvantage of pure breadth-first copying is that, if $x$ points to $y$,
$x$ and $y$ are unlikely to be adjacent after copying.  If they are in different
\label{semi-depth-first}
cache blocks, then memory locality will suffer.  In contrast, depth-first search
improves cache locality in traversing data structures.  A reasonable compromise
is to use mostly breadth-first copying with \emph{local depth-first copying},
which uses an auxiliary stack of bounded depth.  Lines 21-24 of
\autoref{list:forward} implement this.  We have proved it correct, but we have
not done extensive measurements of its performance.

\subsection{The \lstinline{forward} Function}

The \lstinline{forward} function performs the core operation of
Cheney’s algorithm. Its implementation is shown
in \autoref{list:forward}.
\begin{description}
    \item[Line 2:] \lstinline{next} is a pointer-to-pointer so that
    it can be advanced as illustrated in \autoref{fig:cheney}.
    \item[Line 4:] To test whether the last bit is 0 or 1, indicating
    pointer or nonpointer (respectively), we must use
    the function \lstinline{is_ptr}, which in turn uses
    \lstinline{test_int_or_ptr}; see \autoref{tag-bit-semantics}
    for an explanation of that and of \lstinline{int_or_ptr_to_ptr} and
    \lstinline{ptr_to_int_or_ptr}.
    \item[Line 6:] The \lstinline{Is_from} function needs to be
    specially axiomatized; see \autoref{sec:rangetest}.
    \item[Line 7:] \lstinline{Hd_val} indexes \lstinline{v[-1]} to 
fetch the object descriptor, and at line 11 \lstinline{Wosize_hd} extracts the size field from the descriptor.
    \item[Line 8:] The test \lstinline{hd==0} detects the
     $\star$ mark illustrated in \autoref{fig:cheney}, and
     line 18 places the $\star$ mark.
     \item[Line 19] places the \includegraphics{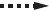}
     forwarding pointer.
    \item[Line 21:] If \lstinline{depth>0} this performs a local depth-first
    traversal (to improve locality of reference).  The default setting
    \lstinline{depth=0} disables this feature, but a value of
    $\sim \!\! 10$ would also be reasonable.
\end{description}

\begin{lstlisting}[language=C,numbers=left,numberstyle=\tiny,float,caption={The \textsf{forward} function},label=list:forward]
void forward (value *from_start,  value *from_limit,  
	      value **next, value *p, int depth) {
  value *v,  va = *p;
  if (is_ptr(va)) {  $~~$ /* $\mbox{see \autoref{tag-bit-semantics}}$ */
    v = (value*)int_or_ptr_to_ptr(va);
    if (Is_from(from_start, from_limit, v)) { $~~$ /* $\mbox{see \autoref{sec:rangetest}}$ */
      header_t hd = Hd_val(v);
      if (hd == 0) { 
        *p = Field(v,0); /* already forwarded, follow forwarding pointer */
      } else {  /* copy to to-space and install forwarding pointer */
        int i, sz=Wosize_hd(hd);
        value *newv = *next+1;
        *next = newv+sz;
        Hd_val(newv) = hd;
        for(i = 0; i < sz; i++) {
          Field(newv, i) = Field(v, i);
        }
        Hd_val(v) = 0;
	Field(v, 0) = ptr_to_int_or_ptr((void *)newv);
	*p = ptr_to_int_or_ptr((void *)newv);
        if (depth>0)
  	 if (!No_scan(Tag_hd(hd)))
            for (i=0; i<sz; i++)
              forward(from_start, from_limit, next, &Field(newv,i), depth-1);
} } } }
\end{lstlisting}
\subsection{Exterior and Interior Pointers}
\label{sec:ext-int}

\begin{figure}
    \begin{minipage}{2in}
    \caption{Does this graph have 5 vertices, or 10?  That is, should the boxes labeled ``root'' or ``outlier'' be considered as graph vertices, or should they bey nonvertices whose contents are addresses of vertices?  The dashed oval labeled ``graph'' hints at our design choice: root boxes and outliers are not graph vertices.}
    \label{fig:interior}
    \end{minipage}\hspace*{1em}\begin{minipage}{3in}
    \includegraphics[scale=1]{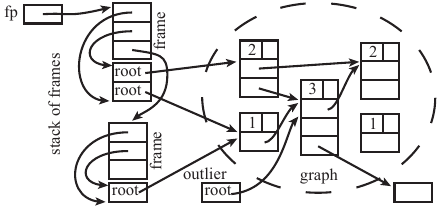}
    \end{minipage}
\end{figure}

The \lstinline{forward} function's implementation
works regardless of how the mutator's root pointers
are organized, regardless of how many generations
the collector has (or whether it's generational at all),
and regardless of how ``remembered sets'' are organized
for mutable references.  Therefore we seek 
a formal specification of \lstinline{forward} 
(that is, a funspec) that's independent of these issues
as well.

As we have discussed, the mutator manages a set of root pointers.  When it calls the collector, the mutator must indicate the locations where it stores those roots.  We use a
stack-of-frames data structure (\autoref{fig:interior}), but the same considerations
would apply regardless of the particular data structure used.
In the specification of the interface between mutator and collector, should we consider that each ``root'' box in the
stack of frames is a graph node (with out-degree 1),
or should we say that they are outside the graph, pointing in?
Each of these design choices has its disadvantages:
\begin{itemize}
    \item If root-boxes are graph vertices, then whenever the mutator copies a root-pointer, its correctness proof needs to create a new graph vertex; and when the root is no 
    longer active, destroy the vertex.  This is quite a bit of bookkeeping.  In such a model, the ``managed g.c. heap'' contains a subset of the ``graph,'' and there are parts of the graph that the g.c. does not manage. 
    \item If root-boxes are not vertices, then our specification must distinguish between \emph{interior} pointers (from one vertex to another) and \emph{exterior} pointers (from outside the graph to graph vertices).  In this model, the ``managed g.c. heap'' is corresponds 1-1 to the graph + allocable space.  \emph{Our specification is based on this principle.}
\end{itemize}
In order to ease the bookkeeping burden for those who prove the correctness of mutators (clients of our g.c.), we have chosen the second way: root-boxes are not vertices,
and we have two kinds of pointers, interior and exterior. 

Sometimes the same C code must be understood in two
different ways depending on whether a pointer is
interior or exterior.  For example, in line 3 of \autoref{list:forward}, the variable \lstinline{va}
receives the \lstinline{value} stored at the memory location pointed
to by \lstinline{p}. 
To prove correctness of this line of code, there must be
a spatial maps-to predicate in the precondition,
giving access to the resource \lstinline{p}.
If \lstinline{p} is an interior pointer (i.e., a field
of a record in the graph \lstinline{g}), we can derive
that resource from the \lstinline{graph_rep g} predicate.
If \lstinline{p} is exterior, we can derive it from
the set of roots---ultimately, from the \lstinline{frames_rep}
predicate that describes the stack-of-frames.  
But the formal specification of
\lstinline{forward} should not need to know the
details of how the stack-of-frames is organized,
since it is passed the address of one root (or 
one field of a graph record) at a time.

Thus, in our specification of the \lstinline{forward} function
we describe the pointer \lstinline{p} with two cases:
\begin{lstlisting}
Definition forward_p_type :=  
    (* exterior *) (Z + GC_Pointer + VType) 
  + (* interior *) VType*Z
\end{lstlisting}
The first case
indicates the three possible content types of exterior pointers
(unboxed integer, outlier pointer, or graph-vertex) and
the second case indicates an interior pointer as the $i$th
field of graph-vertex $v$. Then we need
predicate \lstinline{forward_p_compatible} (not shown) to indicate
compatibility of $p$ with the graph, in either of those cases.

The separation logic representation-predicate \lstinline{forward_p_rep} reflects this
distinction: it represents a one-word resource for exterior pointers,
and is \lstinline{emp} for interior ones, whose memory is already
covered by \lstinline{graph_rep}.
\begin{lstlisting}
Definition forward_p_rep (sh: share) (p: forward_p_type) addr g :=
  match p, addr with
  | FwdPntIntr _ , _ => emp
  | FwdPntExtr extr , Some v => data_at sh int_or_ptr_type (exterior2val g extr) v
  | _ , _ => emp (* this case cannot occur *)
  end.
\end{lstlisting}

\subsection{Specification of the forward function}
The \lstinline{forward} function can be used in a wide variety
of garbage collectors.  In fact, the OCaml g.c.\ has a very similar
function that was almost unchanged\footnote{Three changes: (1) In the OCaml5
collector, checking the mark on an object (as at lines 7-8 of \autoref{list:forward}) must be a
compare-and-swap synchronized operation, to ensure
that two threads do not both copy the object; 
and this must be properly synchronized with reading the forwarding
pointer (as at line 9).  (2) OCaml5 has a new ``continuation'' tag, 
related to stacks of other threads.  (3) The depth-first
``to-do list'' of nodes remaining to forward is now in
a per-thread state structure rather than in a global variable.  Our forward function does not have a 
to-do list in the style of OCaml 4 or OCaml 5.}
in a radical overhaul of the collector
between OCaml 4 and OCaml 5 \cite{ocaml5gc:icfp}.  Furthermore,
though both OCaml 4 and OCaml 5 have 2-generation collectors
with the older generation using mark-and-sweep, almost
the same \lstinline{forward} function can be used in a multi-generation
collection in which all generations use copying collection.

\begin{lstlisting}[float,caption={Funspec of the \lstinline{forward} function},label=list:forward-spec]
DECLARE _forward
WITH rsh: share, sh: share, gv: globals,  g: LGraph, h: heap, hp: val, 
     outlier: outlier_t, from: nat, to: nat, depth: Z, forward_p: forward_p_type,
     fwd_addr: forward_addr_type
PRE [tptr int_or_ptr_type,    $\hspace{2em}$  (* value *from_start *)
     tptr int_or_ptr_type,    $\hspace{2.3em}$  (* value *from_limit *)
     tptr (tptr int_or_ptr_type),  (* value **next *)
     tptr int_or_ptr_type,  $\hspace{2.3em}$ (* value *p *)
     tint]       $\hspace{8.8em}$ (* int depth *)
  PROP (readable_share rsh; writable_share sh;
         graph_heap_compatible g h;
         outlier_compatible g outlier;
         forward_p_compatible forward_p outlier g from;
         forward_p_addr_match forward_p fwd_addr;
         forward_condition g h from to;
         0 <= depth <= Int.max_signed;
         from <> to)
  PARAMS (gen_start g from;       $\hspace{8.4em}$ (* from_start *)
           limit_address g h from;   $\hspace{7em}$ (* from_limit *)
           heap_next_address hp to;  $\hspace{6.2em}$ (* next *)
           forward_p_address forward_p fwd_addr g;  (* p *)
           Vint (Int.repr depth))   $\hspace{7.5em}$  (* depth *)
  SEP (outlier_rep outlier;
        forward_p_rep sh forward_p fwd_addr g;
        graph_rep g;
        heap_rep sh h hp)
POST [tvoid]
  EX g': LGraph, EX h': heap,
  PROP ((g', h') = forward_graph_and_heap from to (Z.to_nat depth)
                         (forward_p2forward_t forward_p g) g h)
  RETURN ()
  SEP (outlier_rep outlier;
        forward_p_rep sh (upd_fwd from to g forward_p) fwd_addr g';
        graph_rep g';
        heap_rep sh h' hp).
\end{lstlisting}
To support this claim that \lstinline{forward} is usable in
many different kinds of garbage collectors, we will show its
specification and proof without (yet) describing the specification
of other components of the g.c.

\autoref{list:forward-spec} shows the VST separation logic function
specification of the \lstinline{forward} function.  
It takes graph \lstinline{g} compatible with heap \lstinline{h} in the
precondition, and produces \lstinline{g'} compatible
with \lstinline{h'} in the postcondition.

Beyond the distinctions arising from exterior and interior
pointers---which lead to case-based definitions
of \lstinline{forward_p_type}, \lstinline{forward_p_compatible},
and \lstinline{forward_p_rep}---there are a few additional subtleties in
the specification.
\begin{itemize}
\item The from-space must be one \lstinline{space} in the \lstinline{heap} and the to-space must be a different one; that is
 we are forwarding from $\mathsf{heap}_\mathsf{from}$ to
 $\mathsf{heap}_\mathsf{to}$, where $\mathsf{from}\not= \mathsf{to}$.
\item The precondition \lstinline{forward_condition g h from to}
is defined as this conjunction:
\begin{lstlisting}
enough_space_to_copy g t_info from to
/\ graph_has_gen g from /\ graph_has_gen g to
/\ copy_compatible g /\ no_dangling_dst g
\end{lstlisting}
and the postcondition assures this of the new graph \lstinline{g'}.
``Enough space to copy'' means that the 
available part of the to-space is at least as large as the
unmarked part of the from-space. ``Graph has gen" means that
the generation numbers \lstinline{from} and \lstinline{to}
are legitimate in \lstinline{g}.  ``Copy compatible'' means
that every $\star$ marked vertex in the graph 
(as in \autoref{fig:cheney}) has a
forwarding pointer
\includegraphics{cheney-forwarding-pointer.pdf} installed
that points to a graph vertex in a different generation.
``No dangling destination" means that every edge in the graph
points to a legitimate vertex.
\item The parameter $p$ is the \emph{address} of a
location containing a pointer $q$ that's possibly into the from-space.
If $q$ doesn't point into the from-space, then it should be left
alone (as per the if-statement on line 6 of \autoref{list:forward}), because the space containing
$q$ is not currently being forwarded (or $q$ is an
outlier).
\item The postcondition establishes that
the new graph and heap \lstinline{(g',h')}
are calculated as \lstinline{forward_}\linebreak[1]\lstinline{graph_and_heap}
applied to \lstinline{(g,h)} and some other arguments.
This is essentially the \emph{functional model}
of the C program, defined by case analysis on the four cases:
unboxed, outlier, exterior (root) pointer, interior (edge) pointer.
Our verification employs a common pattern in proving correctness
of low-level imperative programs: prove that the C program
refines a functional model of an algorithm (a function or inductive relation in Rocq),
and then prove that the algorithm accomplishes the high-level
goal (graph isomorphism).  We defer the functional model
proof to \autoref{sec:gc-correct}.
\end{itemize}

\paragraph{Correctness proof(s) of the forward function}

In this paper we do not show Rocq proofs
such as the proof that \lstinline{forward} 
satisfies the funspec in \autoref{list:forward-spec}.
These proofs are available in the CertiGraph github
repo (see \autoref{appendix:repo}), and the proof scripts can be quite long.

The correctness proof of \lstinline{forward} is doubly
long. In effect, that funspec is really the disjunction of
two different specifications, because of the disjunction
in the definition of \lstinline{forward_p_type}.
We have two different proofs of the same C code,
depending on whether the \lstinline{p} parameter
is interior or exterior, and it would not be trivial
to refactor those into a single proof.   In \autoref{sec:ext-int} we
mentioned the disadvantage of treating roots as 
nonvertices---and needing two proofs of the
same C function is a significant disadvantage indeed.
But it is justified by the way it eases the proof
of every client of the g.c., as explained in
\autoref{sec:ext-int}.

The \lstinline{forward} function is called from three sites:
\begin{itemize}
\item From \lstinline{forward_roots} (as in \autoref{fig:cheney}B,C), where the argument \lstinline{p} is always exterior;
\item From \lstinline{do_scan}, (as in \autoref{fig:cheney}D,E),
where \lstinline{p} is always interior;
\item From \lstinline{forward_remset}, which handles the \emph{remembered set} of
mutable locations that have been stored into (as in \autoref{fig:update},
which shows the mutation but not the remembered set).
In this case (see \autoref{sec:forward-remset})
the stored-into location may be either
interior or exterior (i.e., an outlier).
The line of \lstinline{forward_remset} that calls
\lstinline{forward}
needs two different separation-logic proofs
to handle these two cases.
\end{itemize}
This concludes our discussion of the \lstinline{forward}
function; now let's consider other parts of the collector.
\subsection{Forwarding the roots}

The first step in a typical copying g.c.\ is to \emph{forward the roots},
as illustrated in \autoref{fig:cheney}B--C.  As illustrated
in \autoref{fig:frames1}, the roots are contained in a linked
list of records, where each record describes the bounds of an
array of root-pointers.  The \lstinline{forward_roots} function simply traverses
this data structure (\autoref{list:forward-roots}).

\begin{lstlisting}[language=C,float,caption={Forwarding the roots},label=list:forward-roots]
void forward_roots (value *from_start, $\hspace{4.1em}$  /* beginning of from-space */
                   value *from_limit,  $\hspace{4.1em}$  /* end of from-space */
                   value **next,       $\hspace{6.1em}$  /* next available spot in to-space */
                   struct stack_frame *frames) /* list of frames */
 { struct stack_frame *frame = frames;
   value *start; 
   size_t i, limit;
   while (frame != NULL) {
     start = frame->root;
     limit = frame->next - start;
     frame = frame->prev;
     for (i=0; i<limit; i++)
        forward(from_start, from_limit, next, start+i, DEPTH);
   }
}
\end{lstlisting}

The specification of \lstinline{forward_roots} (not shown) is 
similar to that of \lstinline{forward}, except that instead of the conjunct \lstinline{forward_p_rep} for a single root, 
the precondition
contains \lstinline{frames_rep sh fr} for the entire stack of root-frames.
The proposition \lstinline{forward_condition g h from to}
appears in the precondition, and
 \lstinline{forward_condition g' h' from to} is in the postcondition,
 meaning that \lstinline{forward_roots} preserves this important invariant.

\subsection{Scanning the to-space}

\begin{lstlisting}[language=C,numbers=left,numberstyle=\tiny,float,caption={Scanning the to-space},label=list:do-scan]
void do_scan(value *from_start,  /* beginning of from-space */
	     value *from_limit,  /* end of from-space */
	     value *scan,    $\hspace{2em}$    /* start of unforwarded part of to-space */
             value **next)   $\hspace{1.5em}$    /* next available spot in to-space */
{ while(scan < *next) {
    header_t hd = *((header_t*)scan);
    mlsize_t sz = Wosize_hd(hd);
    int tag = Tag_hd(hd);
    if (!No_scan(tag)) {
      int j;
      for(j = 1; j <= sz; j++)
        forward (from_start, from_limit, next, &Field(scan, j), DEPTH);
    }
    scan += 1+sz;
  }
}
\end{lstlisting}

In a typical copying g.c., after the roots are forwarded,
the next step is to scan the to-space (as shown in \autoref{fig:cheney}D--G).
The region between \lstinline{start} and \lstinline{scan} contains
pointers that have already been forwarded (none of them point
into from-space), and between \lstinline{scan} and \lstinline{next}
they have been copied verbatim, so may point into from-space.
The goal of scanning is to make sure all those pointers are forwarded
into to-space (or left alone, if they already pointed to other spaces).
\autoref{list:do-scan} implements this simple algorithm.

\paragraph{No-scan objects}
The OCaml data format has an 8-bit \emph{tag} field in each object's
descriptor, and tag values $\ge 251$ indicate that the fields of the
object do not contain pointers \emph{regardless of their low-order bits}.
This is useful for representing string values.  The test (at line 9)
for \lstinline{No_scan(tag)} skips the forwarding of the fields of
those objects.

\paragraph{Bug, found by verification}
The if-statement at line 9 of \autoref{list:do-scan} 
is to avoid scanning the fields of string objects.
It would be incorrect to scan those fields,
because even numbers would be misinterpreted as pointers.

\autoref{list:forward}, line 22 contains a similar if-statement
for the same reason.  However, in the original C code,
line 22 was missing---which was a bug---and the program was proved correct!
\cite{wang19:oopsla}
Furthermore, the bug was never detected in testing.
How is that possible?

The original correctness proof of the collector \cite{wang19:oopsla}
was done to a specification that was inadvertently too weak:
it permitted string objects only if every eighth byte was odd.
That is, every field of a no-scan object had to be tagged as an
unboxed value, which is not the intent of the no-scan tag
(and would render no-scan objects useless in practice).
The reason the bug was never detected is that no test case
simultaneously exercised \lstinline{depth>0} and no-scan objects.

It was not until we tried to verify a client program that used
packed bytestrings \cite{korkut25:popl} that the weakness of
the specification was detected:  our client program could not be proved
correct, because it used bytestrings that might contain even numbers
at any position.  When we strengthened the g.c.\ specification,
we found that the \lstinline{forward} function
could not be verified until we fixed the bug.

Garbage collectors are notorious for ``heisenbugs,"
that is, bugs that are very difficult to diagnose.  This would
have been one of those heisenbugs, except that we never had to diagnose
it by running a program.  Program verification did exactly the job
that it should have.  But it could only do so because we 
verified the \emph{adequacy of the specification} by 
proving a client program correct.

\paragraph{Funspec of \lstinline{do_scan}}

We will not show the function spec of \lstinline{do_scan},
which is mostly straightforward, except to remark:
\begin{itemize}
\item The parameter \lstinline{scan} must point to an address in the
to-space that holds a well-formatted record.
That is, there must be some vertex $v$ of the graph
whose \emph{label} is the pair $(\mathsf{to},i)$
where $\mathsf{scan}=\mathsf{start}(\mathsf{to})+i\cdot \mathsf{WORD\_SIZE}$.
\item The vertex $v$ must not be $\star$ marked.
\item A standard theorem about Cheney's algorithm is that
\lstinline{scan} will catch up with \lstinline{next}
before we run out of space.  Informally, that's because the
original graph did fit in the from-space, and the available
part of the to-space is at least as big as the from-space.
We express this formally with a predicate
\lstinline{forward_condition};
the proposition \lstinline{forward_condition g h from to}
appears in the precondition of \lstinline{do_scan}, and
 \lstinline{forward_condition g' h' from to} is in the postcondition,
 and it appears in every loop invariant of the proof,
 meaning that \lstinline{do_scan} preserves this important invariant.
\item The postcondition of \lstinline{do_scan} contains the proposition,
\lstinline{do_scan_relation from to $v$ g g'}.  Thus the VST proof
establishes that the C code of \lstinline{do_scan} implements the
functional model defined as a relation in Rocq
(\autoref{list:do-scan-model}).  The proof in
\autoref{sec:gc-correct} uses this relation as part of the
one-generation correctness argument, together with remembered-set
forwarding, root forwarding, and reset.
\end{itemize}

\begin{lstlisting}[float,caption={Functional model of \lstinline{do_scan}. 
 Compare with the C code in \autoref{list:do-scan} and note that it's
 a close transcription into Rocq of what the C program does.},label=list:do-scan-model]
Inductive scan_vertex_for_loop (from to: nat) (v: VType):
  list nat -> LGraph -> LGraph -> Prop :=
| svfl_nil: forall g, scan_vertex_for_loop from to v nil g g
| svfl_cons: forall g1 g2 g3 i il,
    forward_relation
      from to O (interior2forward (InteriorVertexPos v (Z.of_nat i)) g1) g1 g2 ->
    scan_vertex_for_loop from to v il g2 g3 ->
    scan_vertex_for_loop from to v (i :: il) g1 g3.

Definition no_scan (g: LGraph) (v: VType): Prop := 
                NO_SCAN_TAG <= (vlabel g v).(raw_tag).

Inductive scan_vertex_while_loop (from to: nat):
  list nat -> LGraph -> LGraph -> Prop :=
| svwl_nil: forall g, scan_vertex_while_loop from to nil g g
| svwl_no_scan: forall g1 g2 i il,
    gen_has_index g1 to i -> no_scan g1 (to, i) ->
    scan_vertex_while_loop from to il g1 g2 ->
    scan_vertex_while_loop from to (i :: il) g1 g2
| svwl_scan: forall g1 g2 g3 i il,
    gen_has_index g1 to i -> ~ no_scan g1 (to, i) ->
    scan_vertex_for_loop
      from to (to, i)
      (nat_inc_list (length (vlabel g1 (to, i)).(raw_fields))) g1 g2 ->
    scan_vertex_while_loop from to il g2 g3 ->
    scan_vertex_while_loop from to (i :: il) g1 g3.

Definition do_scan_relation (from to to_index: nat) (g1 g2: LGraph) : Prop :=
  exists n, scan_vertex_while_loop from to (seq to_index n) g1 g2 /\
            ~ gen_has_index g2 to (to_index + n).
\end{lstlisting}

\begin{lstlisting}[float,caption={Funspec of \lstinline{do_scan}},label=list:funspec-do-scan]
DECLARE _do_scan
WITH rsh: share, sh: share, gv: globals,  g: LGraph, h: heap, hp: val, 
     outlier: outlier_t, from: nat, to: nat, to_index: nat
PRE [tptr int_or_ptr_type,
     tptr int_or_ptr_type,
     tptr int_or_ptr_type,
     tptr (tptr int_or_ptr_type)]
  PROP (readable_share rsh; writable_share sh;
         graph_heap_compatible g h;
         outlier_compatible g outlier;
         forward_condition g h from to;
         from <> to; closure_has_index g to to_index;
         0 < available_size h to; gen_unmarked g to)
  PARAMS (gen_start g from;
           limit_address g h from;
           (vertex_address g (to, to_index)) + (- WORD_SIZE) 
           heap_next_address hp to)
  SEP (all_string_constants rsh gv;
        outlier_rep outlier;
        graph_rep g;
        heap_rep sh h hp)
POST [tvoid]
  EX g': LGraph, EX h': heap,
  PROP (graph_heap_compatible g' h';
         outlier_compatible g' outlier;
         forward_condition g' h' from to;
         do_scan_relation from to to_index g g';
         heap_relation h h')
  RETURN ()
  SEP (all_string_constants rsh gv;
        outlier_rep outlier;
        graph_rep g';
        heap_rep sh h' hp).
\end{lstlisting}

\section{Generational garbage collection}\label{sec:generational}

The programs and proofs in \autoref{sec:cheney} are modular,
in the sense that they can be repurposed without change in
a variety of different g.c.\ organizations.  In this section
we describe our specific multi-generation collector.

Assume, for the moment, that there are no assignments into mutable
references or arrays, that objects are \emph{immutable}.

We have up to $N_h$ \lstinline{space}s, called generations, numbered from generation 0, the \emph{nursery}.  Objects are created
and initialized in the nursery;
it is trivial that all the fields of nursery objects point to objects
in generations at least as old as the object itself.
That is: Objects in generation $g_i$ 
point only into generations $g_j$ for $j\ge i$).

A garbage collection entails: 
use Cheney's algorithm to copy all the objects from $g_0$
into the available space at the end of $g_1$
(at which point $g_0$ is empty);
and then either stop, or copy all the objects from $g_1$
into the available space of $g_2$, and so on.
Clearly this will preserve the invariant that
an object in generation $g_i$ will point only to objects
in generations $g_j$ such that $j \ge i$.

This is more efficient than Cheney's original
two-space scheme because:
\begin{itemize}
    \item Generation $i$ is smaller than generation $i+1$,
      so there are fewer objects to examine than if the 
      entire live graph is copied, and usually
      the collector will stop after the $i=0$ iteration;
    \item The \emph{infant mortality} hypothesis (more commonly
    known as the \emph{weak generational hypothesis}) is
    that young objects are likely to die sooner than objects
    that have stayed reachable for a long time already
    \cite{lieberman83}, so it's worth concentrating the
    effort on the youngest objects;
    \item The nursery can fit into the CPU cache,
    since it is all a contiguous region of about that size,
    so its allocation and nursery-collection are efficient \cite{goncalves95:fpca}.
\end{itemize}
To perform Cheney's algorithm on just a portion of memory,
i.e., generation $g_i$,
we must be able to find all the roots of $g_i$.
Because younger objects point only to older objects,
there will be no roots in generations $>i$;
and all generations $<i$ are empty when we are collecting
$g_i$; so the only roots of $g_i$ are the
local and global variables of the program---that is,
the \emph{stack of frames}.

Now let us remove the assumption that objects are immutable.
If field $f_k$ of an object in generation $g_i$ is stored into,
this may create a \emph{back pointer}, a pointer from an older
generation into a younger generation.  Upon assigning to $f_k$, we keep the
address of $f_k$ in a \emph{remembered set} associated
with $g_0$.  When collecting any generation $<i$, we
must treat $f_k$ as one of the roots of the Cheney
copying collection.  When collecting $g_j$ into $g_{j+1}$,
if the remembered-set item points within $g_j$ we promote it into
the remembered set of $g_{j+1}$.  If $f_k$ pointed into $g_{j'}$
for $j'>j$, then the copy of $f_k$ that's now in $g_{j+1}$ is no
longer a back-pointer, it is a ``normal" pointer to an object in
the same generation or older, so it no longer needs a remembered-set item.\footnote{The criteria for promoting a remembered-set item are simple but nontrivial enough
that we had a bug in the program until verification
revealed it.}

Let \lstinline{spaces[$~$]} be an array
of \lstinline{struct space} records (\autoref{list:struct-heap}).
We will keep the remembered set of  $g_i$ in the 
region between \lstinline{spaces[$i$].limit} and
\lstinline{spaces[$i$].rem_limit}; each word
in that region will be the address of a field that has
been stored into.  Thus, the operation \lstinline{mutable_update}
can be implemented as shown in \autoref{list:mutable_update}.
This works because \lstinline{ti->limit} is the boundary
between the nursery's available region and its remembered set.
If $q$ is an unboxed value (nonpointer) then it cannot
cause a back-pointer and does not need to be added to the
remembered set.

\begin{figure}    
        \caption{Each generation must have enough free space (between \lstinline{alloc} and \lstinline{limit}) for all the allocated+allocable
        objects in all younger generations.  To maintain that
        invariant, if it is violated for $g_k$, one can move
        (i.e., \lstinline{forward}) all objects from $g_0$ to $g_1$, then
        forward $g_1$ to $g_2$, etc., up to $g_k$ into $g_{k+1}$.
        The likelihood of reaching large $k$ is
        exponentially small, since each generation is twice the size of the preceding,
        and most objects do not survive this copying
        because they are not reachable.
        If $g_{k+1}$ doesn't exist, it is created from the
        system \lstinline{malloc}.  If that fails, the
        program must abort.}
        \label{fig:generations}
        \parbox{1.4in}{\fontsize{9}{10}\selectfont The \emph{remembered set} of generation $g_i$ contains
        pointers to fields of objects that \emph{might} contain
        pointers from older generations into $g_i$.  For example, the $g_0$ remembered set contains two entries: the address of        the $y$ field and of the $w$ field.  It is unnecessary but
        harmless to have $y$ in the $g_0$ remembered set.
        The purpose of $w$ in the remembered set is that during
        a collection of just generation $g_0$, this will serve
        as a root of Cheney's algorithm.
        }
    \hspace*{1em}\begin{minipage}{3.5in}
    \hspace*{1em}\includegraphics[scale=1]{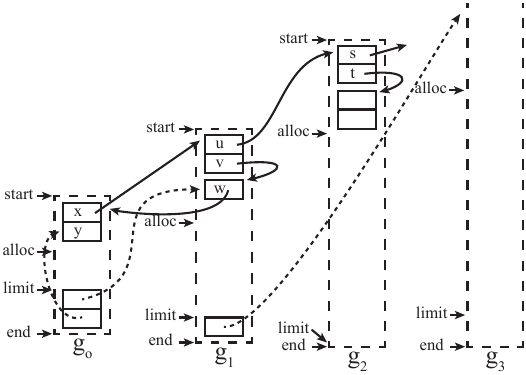}
    \end{minipage}
    \end{figure}
    
\begin{lstlisting}[language=C,float,caption={Storing into a mutable reference},label=list:mutable_update]
void mutable_update(struct thread_info *ti, value *p, value q) {
  *p = q;
  if (is_ptr(q))
    *(value **)(--ti->limit) = p;
}
\end{lstlisting}

The stored-into address \lstinline{p} may be the 
field of a graph object, or it may be an outlier.
The same C function \lstinline{mutable_update}
works in either case, but we need two different
specifications of it; this is a consequence of
the design decision discussed in \autoref{sec:ext-int},
that outliers are not graph nodes.  Storing into an outlier is one way that the mutator can register
a root-pointer that is not a local variable
in the stack of frames.\footnote{However, if
roots are registered in that way, there's no effective
way to \emph{unregister} a root. 
One can store an unboxed integer
into it, so that no graph nodes are retained,
but the address itself cannot be freed because
there will always be a remembered-set item pointing to it.
An alternate way for the mutator to register global
roots, such that the global address itself can be
freed when no longer needed, is to make special
frames at the \emph{bottom} of the stack of frames.
When a global root is no longer needed, its special
frame can be removed or cleared.  See 
\url{https://github.com/jpaykin/certicoq-set-library}}

\begin{lstlisting}[language=C,float,numbers=left,numberstyle=\tiny,caption={Forwarding the remembered
set from one generation to the next},label=list:forward_remset]
void forward_remset (struct space *from,  /* descriptor of from-space */
                     struct space *to,    /* descriptor of to-space */
                     value **next)        /* next available spot in to-space */
{ value *from_start = from->start, *from_limit=from->limit, 
        *from_rem_limit=from->rem_limit, *q;
  for (q=from_limit; q != from_rem_limit; q++) {
    value *p = (value*)int_or_ptr_to_ptr(*q);
    if (!Is_from(from_start, from_limit, p)) {
      forward(from_start, from_limit, next, p, DEPTH);
      *(--to->limit) = (value)p;
    }
  }
}
\end{lstlisting}

\subsection{Forwarding the remembered set}
\label{sec:forward-remset}

With a generational garbage collector, the assignment \lstinline{a->$\!$field=q}
could create a \emph{back pointer}, that is, a pointer from an older object
\lstinline{a} to a younger object \lstinline{q}.  If \lstinline{q} is in a
younger generation than \lstinline{a}, then when collecting (only) \lstinline{q}'s generation
it will be necessary to use \lstinline{a->$\!$field} as a root.
The way generational collectors handle this is to put address \lstinline{&a->$\!$field}
into a \emph{remembered set}.  There are many ways of representing the
remembered set as a data structure; Reppy's collector (and ours) uses
the space in each generation between \lstinline{limit} and \lstinline{rem_limit}.

\autoref{list:forward_remset} shows how to use a remembered set
as roots for forwarding, as well as the criteria for promoting
remembered-set items to the next generation's remembered set.  
Let \lstinline{p} be the address of \lstinline{&a->field};
in general, \lstinline{p} may be a mutable reference, a mutable
field of a record, an array element, or an outlier address.  To store
\lstinline{q} into address \lstinline{p}, the mutator calls
\lstinline{mutable_update(ti,p,q)}, passing also the thread-info pointer 
\lstinline{ti}.  This is illustrated in \autoref{fig:update}, which shows
that \lstinline{limit} is decreased, stealing one word from the allocation
arena to be used in the heap-management data structure.
The implementation of \lstinline{mutable_update} (\autoref{list:mutable_update})
shows what has happened: we consider the space (in every generation)
between \lstinline{limit} and \lstinline{rem_limit} to be part
of the heap-management data structure, and 
\lstinline{mutable_update}
simply stores the address p into the generation-0 part of that data structure
(increasing its size by decrementing \lstinline{limit}).
However, if \lstinline{q} is an unboxed integer, there is no
need to remember it because this store cannot possibly
create a back-pointer.

Later, when the garbage collector forwards
generation $i$ into generation $i+1$,
generation $i$'s remembered set must be processed,
by the function \lstinline{forward_remset}
(\autoref{list:forward_remset}).
The loop body in that function is rather simple
(just four short lines), but reasoning about its correctness is
remarkably complex.

\paragraph{Informal explanation of \lstinline{forward_remset}}
We can assume that \lstinline{from} is generation $i$,
\lstinline{to} is generation $i+1$, and all the generations
$0$ to $i-1$ are empty (because they would have to be collected before
the generational collector would consider forwarding generation $i$).
In \autoref{list:forward_remset}, lines 4--5 load the description
of from-space boundaries into local variables for convenient access.
The for-loop at line 6 processes each element of the remembered set.
We know all of those elements are pointers, not unboxed integers,
since they were the destination arguments of calls to \lstinline{mutable_update}.
Therefore line 7 can use \lstinline{int_or_ptr_to_ptr} to cast the \lstinline{value}
to type \lstinline{value*} (see \autoref{int-or-ptr}).
The pointer, a member of the remembered set, \lstinline{p} is the address of some field that has
been stored into.   
\begin{itemize}
\item If that field is in the current from-space (generation $i$), where could that field point to?  It must point to generation $\ge i$, since
all the younger generations are empty at present.  Therefore it cannot
be a back-pointer, and we can ignore it.
\item If that field is not in the current from-space, then it must be
in a generation $>i$ or be an outlier.
Either way, it's possible that it points into generation $i$,
in which case it would be a back-pointer.  We \lstinline{forward} the
pointer, which has no effect unless the field does indeed point into
generation $i$ (and is harmless otherwise).
On the other hand, perhaps that field points into generation $i+k$;
in that case, this call to \lstinline{forward} does nothing, but
we must preserve \lstinline{p} for later consideration
when collecting generations $>i$.  Line 10 does that by copying
\lstinline{p} into the remembered set of generation $i+1$.
\end{itemize}

\paragraph{Verification Complexity}
The simplicity of lines 4--7 in \lstinline{forward_remset} 
belie the complexity of verifying them.  Some of the issues that
complicate the reasoning are:
\begin{itemize}
    \item The test \lstinline{Is_from(from_start,from_limit,p)} is
    exactly the comparison \lstinline{from_start$\le$p<from_limit}.
    Our axiom for range comparisons (see \autoref{sec:rangetest}) requires us
    to prove that \lstinline{from_start} and \lstinline{from_limit}
    are allocated and in the same block, and that \lstinline{p} is
    allocated.  VST supports the ability to do such proofs, but
    it still requires several lines of Rocq proof script.

    \item Pointer arithmetic in the loop test condition (in the semantics
    of CompCert and Verifiable C) involves
    arithmetic modulo $2^\mathtt{word\_size}$, for which we do not have
    good proof automation.

    \item Our g.c. proof maintains \emph{many} invariants
    about the state of the graph and heap, and about the
    start, limit, and rem\_limit of each generation.  Reestablishing
    these invariants takes several lines of proof script.

    \item Recall from \autoref{sec:ext-int} that the specification
    of the \lstinline{forward} function is split into two cases,
    one for forwarding an interior pointer and one for an exterior pointer,
    and calls to \lstinline{forward} for these cases have nontrivially different
    proof obligations---different treatments of memory layout and separation-logic framing.  The call to \lstinline{forward} from 
    \lstinline{forward_remset} could be in either case; that is,
    \lstinline{p} could be an outlier (exterior) or in a generation
    (interior).  Essentially, we need two proofs of the same loop body,
    one for each such case.  
\end{itemize}

Consequently, we find that the four-line loop body requires
452 lines of Rocq/VST proof script.

\subsection{Proof of generational garbage collection}

\paragraph{Bug fixed in verification}
Our original implementation of \lstinline{forward_remset} had a bug
at line 13: there was \lstinline{q} instead of \lstinline{p}.
This was not found in testing, because we had not yet had test cases
with enough generations \emph{and} with the right pattern of mutable-ref updates.
Instead, we found it (and fixed it) while proving the correctness of the algorithm.

\paragraph{The \lstinline{do_generation} function}

With all of that established, we can present
the \lstinline{do_generation} function for copying the live objects from
$g_i$ to $g_{i+1}$, using both the remembered set
and the stack-of-frames as roots (\autoref{list:do-generation}).

\begin{lstlisting}[language=C,float,caption={Copying one generation into the available space of another},label=list:do-generation]
void do_generation (struct space *from,  /* descriptor of from-space */
                    struct space *to,    /* descriptor of to-space */
                    struct stack_frame *fr)  /* where are the roots? */
{ value *p = to->next;
  forward_remset(from, to, &to->next);
  forward_roots(from->start, from->limit, &to->next, fr);
  do_scan(from->start, from->limit, p, &to->next);
  from->next=from->start;
  from->limit=from->rem_limit;
}
\end{lstlisting}

\paragraph{The \lstinline{garbage_collect} function}

\begin{lstlisting}[language=C,float,numbers=left,numberstyle=\tiny,caption={The \lstinline{garbage_collect} function},label=list:garbage-collect]
void garbage_collect(struct thread_info *ti)   {
  struct heap *h = ti->heap;
  int i;
  h->spaces[0].limit = ti->limit;
  h->spaces[0].next = ti->alloc;
  for (i=0; i<MAX_SPACES-1; i++) {
      struct space *r = h->spaces+i;
      struct space *s = h->spaces+(i+1);
      if (s.start==NULL) {
        int n = 2 * r->rem_limit-r->start; 
        value *p = (value *)malloc(n * sizeof(value));
        if (p==NULL) abort_with ("Could not create the next generation\n");
        s->start=p;
        s->next=p;
        s->limit = p+n;
        s->rem_limit = s->limit;
      }
      do_generation(h->spaces+i, h->spaces+(i+1), ti->fp);
      if (r->rem_limit - r->start <= s->limit - s->next) {
         value *lo = h->spaces[0].start;
         value *hi = h->spaces[0].limit;
         if (hi-lo < ti->nalloc) 
           abort_with ("Nursery is too small for function's num_allocs\n");
         ti->alloc = lo;
         ti->limit = hi;
         return;
      }
    }
  abort_with("Ran out of generations\n");
}
\end{lstlisting}

\autoref{list:garbage-collect} embodies our policy for managing
the generations.
To match the size of
the a typical CPU cache,
we set \textsc{nursery\_size}=$2^{16}$ words or $2^{16\cdot W}$ bytes,
where $W=8$ is the word size in bytes.
We let generation $g_{i+1}$ be twice the size of generation $g_i$.
Since all the generations must fit in an address space
of $2^{W*8}$ bytes, one can calculate
the maximum number of generations that could ever be
usable:
\[
\mathrm{MAX\_SPACES} = (8W - (2+\log_2 (W\cdot\mathrm{NURSERY\_SIZE})))
\]
which is 43 when $W=8$.

For every $i$,
we want to guarantee the \emph{enough space to copy}
property between $g_i$ and $g_{i+1}$.
That means there must be enough available space in $g_{i+1}$
for the entire size of $g_i$.  If that not
the case immediately after a copying collection into the
available space of $g_i$, we perform Cheney's algorithm to copy
all the records out of $g_i$ into $g_{i+1}$.
If that (momentarily) violates 
\emph{enough space to copy}
between $g_{i+1}$ and $g_{i+2}$, then
we copy all records out of $g_{i+1}$, and so on.

If, in the process, we come to a generation $g_k$ that
has not yet been created (where $k<\mathrm{MAX\_SPACES}$),
then the g.c.\ creates that generation by calling
\lstinline{malloc}.  If \lstinline{malloc} returns 0
(failing to find enough virtual memory), the entire program (collector and
mutator) must abort.

We will never reach line 30 $(k\ge \mathrm{MAX\_SPACES})$,
because malloc cannot ever find more virtual memory than
the size of the 64-bit (or $8W$-bit) address space,
so instead it will abort at line 13. That's obvious,
but we have not actually proved it.  

If the mutator wants to create a very large object,
an array or string larger than the size of the nursery,
it will set \lstinline{ti->nalloc} to that large size
and call \lstinline{garbage_collect}, which will abort
at line 24.  A more sophisticated implementation would avoid aborting
by collecting all generations whose size is less than \lstinline{nalloc},
and then allocating the large object directly into one of the older
generations.

\paragraph{Correctness of \lstinline{garbage_collect}}
The VST specifications and proofs of the \lstinline{mutable_update},
\lstinline{do_generation}, and \lstinline{garbage_collect}
are expressed in similar manner to that of the funspecs shown earlier
in this \autoref{sec:cheney}.  Like those proofs, we prove that the
C programs implement functional models of algorithms expressed
as functions and inductive relations in Rocq.  Then we prove the
algorithms correct.

\section{Correctness of the functional model}\label{sec:gc-correct}

In Sections~\ref{sec:cheney} and~\ref{sec:generational}, the VST
proofs establish that the C code implements functional models written
in Rocq: for example, \lstinline{forward} implements
\lstinline{forward_graph_and_heap}, and \lstinline{do_scan} implements
\lstinline{do_scan_relation}.  This section shows how the functional models of all the C functions are proved to satisfy top-level specification:
that they preserve the graph structure visible to the mutator, up to isomorphism.

To establish this result, we first state the graph-isomorphism property
used as the mutator-facing conclusion.  We then prove it for
one-generation copying and present the remembered-set invariant needed
for mutable update.  Finally, we lift the one-generation result to the
full \lstinline{garbage_collect} loop.

In the theorem below, \lstinline{roots1} and \lstinline{roots2} are the
ordinary, mutator-visible root lists extracted from the stack of frames
described in \autoref{sec:interface}.  The additional roots used
internally during a one-generation collection are introduced in
\autoref{subsec:copying-one-generation}.

\subsection{Graph isomorphism for the root-reachable subgraph}
\label{subsec:root-reachable-isomorphism}

The predicate \lstinline{gc_graph_iso} formalizes the graph-isomorphism
property stated informally in \autoref{sec:spec}.  Root boxes
themselves are not vertices of the graph, as discussed in
\autoref{sec:ext-int}.  Therefore the definition first extracts from
each root list the roots that actually point to graph vertices,
restricts each graph to the subgraph reachable from those vertices,
and then requires an explicit label-preserving graph isomorphism
between the two reachable subgraphs.

\begin{lstlisting}
Definition gc_graph_iso (g1: LGraph) (roots1: roots_t)
                        (g2: LGraph) (roots2: roots_t): Prop :=
  let vertices1 := filter_proj exterior_proj_vertex roots1 in
  let vertices2 := filter_proj exterior_proj_vertex roots2 in
  let sub_g1 := reachable_sub_labeledgraph g1 vertices1 in
  let sub_g2 := reachable_sub_labeledgraph g2 vertices2 in
  exists vmap12 vmap21 emap12 emap21,
    roots2 = map (exterior_map vmap12) roots1 /\
    label_preserving_graph_isomorphism_explicit
         sub_g1 sub_g2 vmap12 vmap21 emap12 emap21.
\end{lstlisting}

The map on roots accounts for copying.  A root that points to a graph
vertex before collection may point to a different vertex afterwards,
namely the copy in another generation.  The
\lstinline{exterior_map} function maps graph-vertex roots through
\lstinline{vmap12} and leaves unboxed integers and outlier pointers
unchanged.

The main graph-isomorphism theorem has the following form:

\begin{lstlisting}
Theorem garbage_collect_isomorphism:
  forall roots1 roots2 g1 h1 rh1 rmst1 g2 h2 rh2 rmst2,
    graph_unmarked g1 ->
    no_unrecorded_backward_edge g1 rh1 ->
    no_dangling_dst g1 ->
    roots_graph_compatible roots1 g1 ->
    sound_gc_graph g1 ->
    remset_graph_state g1 O rmst1 rh1 ->
    remset_heap_covers_graph g1 rh1 ->
    garbage_collect_relation roots1 roots2 g1 h1 rh1 rmst1 g2 h2 rh2 rmst2 ->
    gc_graph_iso g1 roots1 g2 roots2.
\end{lstlisting}

Here \lstinline{g1} and \lstinline{g2} are the pre- and post-collection
graphs, and \lstinline{h1} and \lstinline{h2} are the corresponding abstract
heaps.  The variables \lstinline{rh1} and \lstinline{rh2} are the
remembered-set heaps, while \lstinline{rmst1} and \lstinline{rmst2} are the
remembered-set states used to relate those heaps to the graph.  These
components are needed to state the collector's internal invariants,
but they do not occur in the conclusion.  The conclusion mentions only
the two graphs and the two root lists.

The theorem also assumes standard graph well-formedness conditions.
In particular, \lstinline{sound_gc_}\linebreak[1]\lstinline{graph} says that the graph's
validity predicates agree with its represented vertex and edge sets,
and that each edge's source and label agree with the two components of
its representation.  More explicitly, \lstinline{graph_unmarked}
requires all graph vertices to be unmarked,
\lstinline{no_dangling_dst} requires every represented edge to have a
valid destination, and \lstinline{roots_graph_compatible} requires
every vertex-valued root to denote a valid graph vertex.  The
hypotheses \lstinline{no_unrecorded_}\linebreak[1]\lstinline{backward_edge},
\lstinline{remset_graph_state}, and
\lstinline{remset_heap_covers_graph} express the remembered-set
discipline discussed in
\autoref{subsec:remembered-sets-recorded-back-pointers}.  Separate
preservation lemmas show that the graph operations used by the
collector maintain these properties.

The hypotheses of \lstinline{garbage_collect_isomorphism} are
deliberately stated as separate graph-level facts rather than through
the bundled VST predicates used at the function boundary.  In the full
current funspec, bridge lemmas derive these facts from
\lstinline{full_gc} and \lstinline{remembered_set_ok}, together with
facts supplied by the spatial representation predicates.  The full VST
precondition is stronger: it also contains heap-layout and copying-safety
conditions needed to verify the C implementation but not used directly
by the graph-isomorphism argument.

For readability, the funspec displayed in
\autoref{list:gcspec} is abridged: it suppresses the remembered-set
component of the precondition and the heap and remembered-set arguments
of \lstinline{garbage_collect_}\linebreak[1]\lstinline{relation}.  The full
relation used by the correctness theorem threads the abstract heap,
remembered-set heap, and remembered-set state through the collection.
These additional arguments are needed for the collector's internal
invariants; the resulting \lstinline{gc_graph_iso} conclusion remains
the graph-isomorphism property described in \autoref{sec:spec}.

\subsection{Copying one generation}
\label{subsec:copying-one-generation}

The proof is first developed for one generation.  The C function
\lstinline{do_generation}, shown in \autoref{list:do-generation}, copies
live objects from generation \(g_i\) into the available space of
generation \(g_{i+1}\).  The functional-model relation follows the same
sequence as the C program:

\begin{center}
\lstinline{forward_remset};\hspace{1em}
\lstinline{forward_roots};\hspace{1em}
\lstinline{do_scan};\hspace{1em}
\lstinline{reset}.
\end{center}

The pointer \lstinline{p} saved at the beginning of
\autoref{list:do-generation} marks the start of the part of the
to-space that must later be scanned.  The calls to
\lstinline{forward_remset} and \lstinline{forward_roots} may copy objects
into that to-space and advance \lstinline{to->next}.  Then
\lstinline{do_scan} performs the Cheney closure by scanning the objects
copied between \texttt{p} and the final value of \lstinline{to->next}.
Finally, the from-generation is reset by restoring \lstinline{from->next}
and \lstinline{from->limit}, making the generation available for future
allocation.

We call the local and global variables represented by the stack of
frames described in \autoref{sec:interface} the ordinary roots.  In the
immutable case, one could explain the one-generation proof using only
these roots.  With mutable update, that is not enough.  As
\autoref{sec:generational} explains, a stored-into field in an older
object may point into the generation currently being collected.  That
older object is not itself in the from-space, and therefore it will not
be scanned during this collection.  The remembered set records the
address of such a field, and \lstinline{forward_remset} uses that
address as an additional root-like source for this collection step.

The proof makes this distinction explicit by considering an augmented
root list.  At the C level, \lstinline{forward_remset} processes
remembered-set locations.  At the graph level, the corresponding
additional roots are the target vertices currently stored at those
locations.  The augmented list is obtained by prefixing the
ordinary-root list with the targets of the entries that
\lstinline{forward_remset} processes for generation \(g_i\): entries
whose source locations are outside the current from-space and whose
current targets lie in \(g_i\).  Thus these remembered-set targets form
the prefix, while the ordinary stack-of-frames roots form the suffix.

The one-generation argument first establishes the graph-isomorphism
property for this augmented root list.  This corresponds to the
algorithmic role of \lstinline{forward_remset}: before the collector
forwards the ordinary roots and scans the to-space, it accounts for the
additional root-like sources represented by the remembered set.

The main technical step relates reachability from the augmented roots
in the original graph to the reachable-or-marked vertices after
remembered-set forwarding.  After the ordinary roots have been
forwarded and the to-space has been scanned, every vertex in generation
\(g_i\) reachable from the augmented roots has been forwarded and 
is marked, and no edge in the scanned graph points into \(g_i\).  
These two facts justify resetting the collected generation and yield
the graph-isomorphism property for the augmented root list.

The graph-isomorphism result for the augmented root list is only an
intermediate statement.  The one-generation result used by the outer
\lstinline{garbage_collect} loop is stated for the ordinary
stack-of-frames roots.  To obtain it, the proof restricts the
augmented-root isomorphism back to the ordinary roots.  This uses a
general property of \lstinline{gc_graph_iso}: an isomorphism between
root lists with corresponding equal-length prefixes can be restricted
to the suffixes consisting of the ordinary roots.  After rewriting the
post-collection root list so that the mapped remembered-set roots form
the corresponding prefix, these additional roots no longer appear in
the \lstinline{gc_graph_iso} conclusion used by the top-level collection
theorem.

This restriction does not undo forwarding or discard any object
reachable from the ordinary roots.  It changes only the root set from
which the two graphs are observed.  A target vertex that is reachable
from an ordinary root remains in the ordinary root-reachable subgraph,
even if it is also reached through a remembered-set location.  Other
target vertices may be reached only through the additional
remembered-set roots; they lie outside the restricted subgraph and
therefore do not affect this particular
\lstinline{gc_graph_iso} conclusion.

Thus, once the reachability bridge and the loop invariants have been
established---including graph soundness, root compatibility, the
remembered-set invariants, and the fact that all generations younger
than \(g_i\) are already empty---the functional-model relation for copying
generation \(g_i\) into \(g_{i+1}\) preserves
\lstinline{gc_graph_iso} for the ordinary roots.  The remembered set is
used internally to make the copying step complete, but it does not
appear in the resulting root-reachable graph-isomorphism conclusion.

\subsection{Remembered sets and recorded back pointers}
\label{subsec:remembered-sets-recorded-back-pointers}

The main proof invariant added for mutable update is a precise version
of the remembered-set discipline described in
\autoref{sec:generational}.  A back pointer into generation \(g_i\)
need not already be recorded in the remembered set of \(g_i\).  The
write barrier first records the source field in the remembered set of
generation \(g_0\), and the collector moves such remembered-set
entries into \(g_1\), then \(g_2\), etc., as it proceeds through the generations.  Thus the invariant
must describe not just where an entry ultimately matters, but where it
may be in the collector's current remembered-set pipeline.

In the graph model, pointer-bearing object fields are represented as
graph edges.  An edge \(e=(v,j)\) identifies the \(j\)-th field of
source vertex \(v\); its destination is the vertex currently stored in
that field.  Thus, \lstinline{remset_heap_records_edge rh k e} means
that the location of this source field occurs in the \(k\)-th component
of \lstinline{rh}.  The remembered set records the source location, not
the destination vertex:

\begin{lstlisting}
Definition remset_heap_records_edge (rh: remset_heap) (k: nat) (e: EType): Prop :=
  In (RemSetInterior (InteriorVertexPos (fst e) (Z.of_nat (snd e)))) (nth_remset_space rh k).
\end{lstlisting}

The indexed recorded-back-pointer invariant is:

\begin{lstlisting}
Definition no_unrecorded_backward_edge_from (from: nat) (g: LGraph) (rh: remset_heap): Prop :=
  forall e, graph_has_e g e ->
      egeneration e > vgeneration (dst g e) ->
      exists k, from $\le$ k $\le$ vgeneration (dst g e) /\ remset_heap_records_edge rh k e.

Definition no_unrecorded_backward_edge (g: LGraph) (rh: remset_heap): Prop :=
  no_unrecorded_backward_edge_from O g rh.
\end{lstlisting}

The comparison
\lstinline{egeneration e > vgeneration (dst g e)} selects a pointer
from an older generation to a younger one.  If the destination of the
edge is in generation \(g_i\), then
\lstinline{no_unrecorded_}\linebreak[1]\lstinline{backward_edge_from from g rh} says that the
source field of that edge is recorded in some remember\-ed-set component
\(g_k\), where \(\mathit{from} \leq k \leq i\).  At the public entry point of
\lstinline{garbage_collect}, the collection starts from \(g_0\), so
the unindexed name \lstinline{no_unrecorded_backward_edge} is the
special case with \(\mathit{from} = 0\).

This predicate concerns graph-internal edges, whose source locations
are represented by \lstinline{RemSet}\-\lstinline{Interior}.  Remembered-set
locations in outlier objects are not graph edges; they are covered
separately by the remembered-set compatibility conditions used at the
VST interface.

This condition replaces the simpler invariant available in the
immutable setting.  Without mutable update, once an object is created
in the nursery, its fields point only to objects in the same or older
generations, and older objects cannot later acquire pointers to
younger objects.  With \lstinline{mutable_update}, that statement is
false.  The right invariant is not that back pointers do not exist,
nor that their source fields are already recorded in the destination
generation's remembered set.  Instead, every graph-internal
old-to-young edge must be recorded somewhere in the remembered-set
pipeline between the generation currently being collected and the
edge's destination generation.

The recorded-back-pointer condition is kept separate from ordinary
graph/heap compatibility.  Ordinary compatibility predicates say that
the graph is represented in the heap, that roots and outliers are
compatible with the graph, and that edges do not dangle.  The
remembered-set condition talks about the heap-management data
structure, which is opaque to the mutator.  Keeping these conditions
separate makes the specification boundary clearer: remembered sets are
essential to the collector's internal correctness argument, but they
are not part of the graph structure that the mutator sees.

The augmented-root isomorphism described above is already stated for
the graph after generation \(g_i\) has been reset.  A separate
obligation, needed by the next iteration of the collector loop,
concerns resetting the corresponding component of the remembered-set
heap.  This reset must advance the indexed recorded-back-pointer
invariant: before collecting \(g_i\), pending entries may live in
remembered-set components \(g_k\) with \(i \leq k \leq d\), where
\(g_d\) is the destination generation of the edge; after \(g_i\) has
been collected and its remembered-set component reset, the next loop
iteration starts at \(g_{i+1}\), so pending entries must live in
components \(g_k\) with \(i+1 \leq k \leq d\).

Intuitively, every target vertex in \(g_i\) reached from a relevant
remembered-set location has already been forwarded before the graph is
reset.  Any remaining old-to-young edge into a later generation either
already had a witness in a remembered-set component above \(g_i\), or
had a witness in \(g_i\)'s remembered set and is promoted to the
remembered set of \(g_{i+1}\).  Edges whose destinations would be in
generations below \(g_{i+1}\) are ruled out by the invariant that
those generations have already been cleared, together with the
no-dangling-destination property.  Formally, under the graph and
remembered-set invariants maintained by the collector loop,
\lstinline{do_generation_relation} advances
\lstinline{no_unrecorded_backward_edge_from i} to
\lstinline{no_unrecorded_backward_edge_from (S i)} when the
remembered-set component for \(g_i\) is reset.

This preservation theorem is used for the next iteration of the
collector, not just for the current graph-isomorphism result.  The
one-generation isomorphism proof shows that the current collection
preserves what the mutator can observe through the ordinary roots.
The recorded-back-pointer preservation proof shows that the resulting
graph and remembered sets are a valid starting point for later
generation collections.

\subsection{The garbage-collection loop}
\label{subsec:garbage-collection-loop-correctness}

The top-level relation \lstinline{garbage_collect_relation} follows the
C loop shown in \autoref{list:garbage-collect}.  The collector starts
at generation \(g_0\).  At each iteration, it ensures that the next
generation exists, creating it with \lstinline{malloc} if necessary.  It
then calls \lstinline{do_generation} to copy the current generation into
the next one.  If the next generation has enough remaining space, the
collector restores the nursery allocation pointers in the
\lstinline{thread_info} structure and returns.  Otherwise, it proceeds
to the next generation.

The proof of \lstinline{garbage_collect_isomorphism} is by induction
over this functional loop.  The base case is reflexivity of
\lstinline{gc_graph_iso}.  Besides graph isomorphism, the induction
maintains that all generations younger than the one currently being
collected have already been cleared.  This auxiliary invariant holds
initially for generation \(g_0\), and each one-generation collection
advances it to the next generation.  In the inductive case, the proof
has three steps.

First, ensuring that the next generation exists is graph-isomorphism
neutral.  If it already exists, the graph and heap are unchanged;
otherwise, an empty graph generation and its corresponding unused heap
space are added.  In either case, the graph and remembered-set
invariants needed by later steps are preserved.

Second, the roots-isomorphism result for one-generation collection
gives the graph-isomorphism property for the ordinary roots after
\lstinline{forward_remset}, \lstinline{forward_roots},
\lstinline{do_scan}, and reset.  A separate reset-preservation result
advances the indexed recorded-back-pointer invariant from generation
\(g_i\) to generation \(g_{i+1}\), for the remembered-set heap passed
to the next iteration.

Third, if the collector loop continues, the induction hypothesis gives
graph isomorphism for the remaining generations.  Transitivity of
\lstinline{gc_graph_iso} then composes the isomorphisms from ensuring
that the next generation exists, collecting the current generation,
and executing the remaining iterations.

Finally, the theorem
\lstinline{garbage_collect_spec_}\linebreak[1]\lstinline{preconditions_imply_isomorphism}
connects the VST funspec to the graph-level theorem.  Unfolding
\lstinline{full_gc} and \lstinline{remembered_set_ok} provides the
compatibility, collector-entry, and remembered-set facts from which
most of the hypotheses of
\lstinline{garbage_collect_isomorphism} are derived.  The spatial
predicate \lstinline{remset_rep} entails \lstinline{remset_nodup},
while \lstinline{sound_gc_graph} is maintained as a separate semantic
invariant.  The VST body proof establishes
\lstinline{garbage_collect_relation}.  Together, these facts yield
\lstinline{gc_graph_iso} for the mutator's ordinary roots.

Consequently, the C collector may update its internal heap and
remembered sets while preserving, up to isomorphism, the object graph
reachable from the mutator's roots.  This is the mutator-facing
property stated in \autoref{sec:spec}.

\section{Future work}\label{sec:future}

\paragraph {Allocation of large objects}

When the mutator wants to allocate a very large array
(or immutable string object), if it is larger than the size
of the nursery then the call to \lstinline{garbage_collect}
can never satisfy the \lstinline{nalloc} requirement
(enough headroom in the nursery to allocate the object).
What the g.c. should do in this case is collect all generations
smaller than the requested \lstinline{nalloc},
and (temporarily) use generation $k$ as the nursery.
This would be just a few lines of code, but it is easy to imagine
getting it wrong; hence the utility of formal verification.

\paragraph{Safe points}
Multithreaded garbage-collected systems (including OCaml) need
a way to synchronize or suspend threads at \emph{safe points} in their
executions---that is, points where the thread's set of live roots
is understandable by the collector \citep[\S 9.6]{jones23}.  
Line 4 of \autoref{list:ensuring} is such a point: the
root pointers have all been stored into the stack-of-frames.
If a thread were to be interrupted at some other point, 
e.g.\ for the purpose of synchronizing the start of a 
concurrent garbage collection, there is no guarantee that
the g.c.\ could accurately find the roots in local variables.

The mechanism that ML compilers (including SML/NJ and OCaml) use,
when it's desired to stop some or all threads at safe points,
is to set the \lstinline{ti->limit} pointers of those threads
(in each of their \lstinline{thread_info} structs) to the
\lstinline{start} point of their respective allocation spaces.
That way, the next time the thread reaches an 
\lstinline{available} check, it will fail;
the root pointers will be stored into the stack-of-frames
(as in \autoref{list:ensuring}) and \lstinline{garbage_collect}
will be called.  At that point, \lstinline{garbage_collect}
must properly interpret the reason for its invocation:
either because the allocation space is truly exhausted,
or because an interrupt is requested.

For best results, the mutator should perform only a bounded
amount of computation between \lstinline{available} checks.

For correct results, when the mutator fetches
\lstinline{ti->limit}, it must do so by an \emph{atomic load},
i.e., a machine instruction that has the right semantics 
to interact with stores that come from another thread. 
Other mutator operations need not be atomic (because they operate on a per-thread nursery not shared with other threads), except
for the load or store of a mutable reference.

Parallel shared-memory clients
in the style of OCaml4 would be straightforward to implement and
mostly straightforward to verify, though it would require
using the features of VST for reasoning about concurrent memory
\cite{mansky24:iris}.

\paragraph{Verifying a high-performance parallel collector}
A more ambitious project would be to verify the OCaml5
collector \cite{ocaml5gc:icfp}.  One of the design goals of
that collector was to be compatible with the OCaml's legacy
mutator/collector interface.  So it would be interesting to see
which changes would be needed to the \emph{specification
interface} that we described in \autoref{sec:spec}.

\paragraph{Tracing in and out of external objects}
Our theory of external objects is sufficient for
(1) code blocks of closures, (2) abstract objects totally managed
by the runtime system, and (3) external root pointers (into the heap) managed by the runtime system.  But we have not demonstrated that one can reason effectively about traversable data structures that span between the g.c.-managed heap and the runtime-system-managed malloc/free arena.

The OCaml5 forward function is called \lstinline{oldify_one},
and it supports some features that we don't support:
\begin{itemize}
    \item OCaml (4 or 5) has special support for \emph{lazy} objects, 
    which have their own tag.
    \item OCaml objects can be nested within other objects; these are
    called \emph{Infix} objects with a special tag value, and the surrounding object must be found and forwarded.
    \item The queue of objects that have been copied and whose fields have not yet been forwarded, which Cheney's algorithm keeps in the (contiguous) to-space between \lstinline{scan} and \lstinline{next}, cannot be kept in
    that way, because the to-space is not organized that way.  Instead,
    a ``to-do list'' is linked through all the objects that have been
    copied and not yet scanned.  The to-do list is kept depth-first,
    but (unlike the system we described in \autoref{semi-depth-first}) this does not
    necessarily improve cache locality because the to-space is not
    a series of consecutive blocks.
    \newline
    \emph{The following items are specific to OCaml 5.}
    \item Reading the header is done by an acquire-mode atomic load,
    to synchronize in case some other shared-memory thread already started
    forwarding this object.
    \item To allocate the destination block, where our \lstinline{forward}
    function can simply add \linebreak $(1+n)\cdot \mathrm{WORD\_SIZE}$ to
    \lstinline{next}, OCaml5 must find a block of the right size
    in a local pool of free blocks within the shared major generation.
    \item When the destination block has been allocated, the function
    must check \emph{again} whether some other thread is forwarding this
    object (and if so, abandon the newly allocated block as garbage).
    \item Continuations (pointers to the stacks of
    ``fibers'' (lightweight threads within an operating-system
    thread) are tagged with \lstinline{Cont_tag}, 
    and must be handled specially.
    \item The ``to-do list'' is kept in a per-thread state
    structure, rather than in a global variable.    
\end{itemize}
Verifying the OCaml5 collector would be quite challenging
because of the sophisticated
techniques used to deal with concurrency.
\section{Conclusion}\label{conclusion}

We set out to verify a ``real-world'' garbage collector, in the
sense that it is a separate and portable C program with
a standard API usable by multiple client systems;
and where that API could be formalized into a specification
strong enough to prove client programs (mutators) correct.

The proof required significantly more effort than we expected,
perhaps for several reasons:
\begin{itemize}
\item 
As described by \citet{AppelN20}:
\begin{quote}
VST proofs are typically rather verbose (an
order of magnitude more lines of proof than lines of code),
but in this case the proof is particularly lengthy (compared
to the C code). We believe the reason for this is that a malloc/free program is particularly abusive of the C type system:
casting undifferentiated memory into linked data structures,
returning pointers that are one past the header word, making
sure that those pointers are double-word aligned, and so on.
All this abuse is legal in C11, but needs formal justification.
In contrast, the proofs of more “well behaved” C code are a
bit more automated by VST’s proof tactics.  
\end{quote}
The particularly high ratio of proof to code that they described in verifying a malloc/free memory manager also
applies to an automatic garbage collector, for some of the same
reasons.
\item Garbage collectors manipulate graphs
with sharing and (sometimes) cycles; to fit this into separation
logic (where separation is the antithesis of sharing) required
a graph-reasoning system such as CertiGraph.  Such reasoning
also increases the proof-to-code ratio. (Malloc/free systems
may represent graphs with sharing but the verification of malloc/free
itself does not require reasoning about those graphs.)
\item Adding \lstinline{mutable_update}---and in general the handling
of remembered sets in each generation---required changes to many
of the invariants that we had designed for the generational
collector without mutable update.
\item Handling \emph{outliers}, pointers from outside the heap into it,
and vice versa, significantly complicated our reasoning in some places.
Most other garbage-collector proofs prohibit outliers (or are not sound
for the semantics of C pointer comparison).  But many real-life
clients need them, for interaction with runtime systems,
so we chose to provide them.
\item Adjusting the specification of the collector to be usable
in correctness proofs of clients (that call the collector)
required careful refactoring.
\end{itemize}
Clients, too, need to reason about graphs with sharing.  An important
part of the (informal) specification of systems such as OCaml, Haskell, SML/NJ
is that they do not duplicate objects.  That is,
\lstinline{let y : int->int = f x in (y,y)} does not compile to
\lstinline{(f x, f x)} with two function closures,
there is a single function closure with two pointers to it.  So our
g.c. specification also provides the graph-reasoning theory
directly usable in client-compiler proofs.

We chose a multigeneration copying collector because its implementation is
simple and reasonably efficient.  But a more common modern arrangement
(not only in OCaml) is a hybrid two-generation collector,
with copying collection for the nursery and mark-sweep for the
older generation.  Parts of our collector (and proof) should be reusable
in such a scenario, but (as we have explained) there are some significant 
differences.

\begin{acks}
This work was funded in part by National Science Foundation grant CCF-2005545.
\end{acks}
\bibliographystyle{ACM-Reference-Format}
\bibliography{appel}
\clearpage


\appendix

\section{Open-source Rocq proof}\label{appendix:repo}
Our Rocq proofs are available open-source at
\url{https://github.com/CertiGraph/CertiGraph/releases/tag/v2.0}, in the \lstinline{CertiGC} 
directory of that repo, organized as follows:
\begin{description}
\item[GC\_Source/gc\_stack.\{h,c\}:] C source code (including Listings
    \ref{listing-value}, \ref{list:forward},
     \ref{list:forward-roots}, \ref{list:do-scan},
      \ref{list:mutable_update}, \ref{list:do-generation},
    \ref{int-or-ptr}).
\item[gc\_spec.v:] Function specifications of all the C functions
(Listings \ref{list:gcspec}, \ref{list:updatespec}, 
  \ref{list:funspec-do-scan}, \ref{specs-int-or-ptr}).
\item[verif\_forward.v, verif\_forward\{1,2\}.v:] Correctness proof of the \lstinline{forward} function.
\item[verif\_do\_scan.v:] Correctness proof of the \lstinline{do_scan} function.
\item[verif\_forward\_remset.v:] Correctness proof of \lstinline{forward_remset}.
\item[verif\_forward\_roots.v:] Correctness proof of \lstinline{forward_roots}.
\item[verif\_do\_generation.v:] Correctness proof of \lstinline{do_generation}.
\item[verif\_garbage\_collect.v:] Correctness of the client-facing function \lstinline{garbage_collect()}.
\item[gc\_correct.v:] Proof of correctness of the functional model (see \autoref{sec:gc-correct}).
\item[verif\_mutable\_update.v:] Proof of correctness of 
\lstinline{mutable_update}.
\end{description}

To build the garbage collector proof, run \lstinline{make certigc}.
To execute the Rocq \lstinline{Print Assumptions} command that verifies that there are no unexpected axioms or unproved lemmas, run\linebreak \lstinline{make CertiGC/assumptions.vo}.

The ``top-level specification'' or ``main theorem'' is, effectively, the \emph{specification contract} between the mutator and the g.c.  We described this (informally) in \autoref{sec:interface} and (formally) in \autoref{sec:spec}.  To see the complete specification, one would examine in particular,
\begin{itemize}
    \item In \lstinline{gc_spec.v}, the function specs \lstinline{garbage_collect_spec} and \lstinline{int_mutable_update_spec} (which allows the client
    to update a mutable reference that's in the g.c.-managed heap).
    \item In \lstinline{spatial_gc_graph.v}, the definitions of predicates referenced from those two funspecs.
    \item In \lstinline{gc_correct.v}, Theorem 
    \lstinline{garbage_collect_isomorphism} states that the functional model (which the C function \lstinline{garbage_collect} is proved to implement) correctly produces the appropriate isomorphic subgraph; Theorem
    \lstinline{mutable_update_gc_ready}
    states that mutable-ref-update produces the appropriate relation on graphs.
\end{itemize}
\section{The C semantics of tag bits}\label{tag-bit-semantics}

Most implementations of ML (e.g., OCaml, Standard ML of New Jersey) use the low-order-bit tagged 
representation described above:  even numbers are pointers, odd numbers are non-pointer data.
Both the mutator and the g.c.\ must test the low-order bit of a value to determine
how to handle it.  But in the official semantics of C,
this test is undefined behavior!  That is, taking a pointer value
and doing a bitwise-and (or by other means such as modulo) is not permitted
in portable C programs.  When we prove C programs correct, we use the VST program
logic that is proved sound w.r.t. the formal semantics of C, and it would
be impossible to prove such a program correct.

But such programs \emph{are} correct, in the sense that they work correctly in any standard
implementation of C (such as gcc, clang, CompCert).  So we axiomatize an extension
to the C standard that permits \emph{just enough} nonstandard (but commonly supported)
operations to implement a garbage collector.

\begin{lstlisting}[language=C,float,caption={Pointer/integer testing and conversion},label=int-or-ptr]
int test_int_or_ptr (value x) /* 1 if int, 0 if aligned ptr */ { return (int)(((int)x)&1); }

int int_or_ptr_to_int (value x) /* precondition: is int */ { return (int)x; }

void * int_or_ptr_to_ptr (value x) /* precond: is aligned ptr */ { return (void *)x; }

value int_to_int_or_ptr(int x) /* precondition: is odd */ { return (value)x; }

value ptr_to_int_or_ptr(void *x) /* precondition: is aligned */ { return (value)x; }
\end{lstlisting}

\begin{lstlisting}[float,caption={Specifications for pointer/integer testing and conversion},label=specs-int-or-ptr]
Definition test_int_or_ptr_spec :=
 DECLARE _test_int_or_ptr
 WITH x : val
 PRE [int_or_ptr_type]  PROP (valid_int_or_ptr x) PARAMS (x) SEP ()
 POST [ tint ]
   PROP() RETURN(Vint (Int.repr (match x with Vlong _ => 1 | _ => 0 end))) SEP().

Definition int_or_ptr_to_int_spec :=
  DECLARE _int_or_ptr_to_int
  WITH x : val
  PRE [int_or_ptr_type ]  PROP (is_int I32 Signed x) PARAMS (x) SEP ()
  POST [ tlong ]          PROP() RETURN (x) SEP().

Definition int_or_ptr_to_ptr_spec :=
  DECLARE _int_or_ptr_to_ptr
  WITH x : val
  PRE [int_or_ptr_type ]   PROP (isptr x) PARAMS (x) SEP ()
  POST [ tptr tvoid ]      PROP() RETURN (x) SEP().

Definition int_to_int_or_ptr_spec :=
  DECLARE _int_to_int_or_ptr
  WITH x : val
  PRE [ tlong ]              PROP (valid_int_or_ptr x) PARAMS (x) SEP ()
  POST [ int_or_ptr_type ]   PROP() RETURN (x) SEP().

Definition ptr_to_int_or_ptr_spec :=
  DECLARE _ptr_to_int_or_ptr
  WITH x : val
  PRE [tptr tvoid ]         PROP (valid_int_or_ptr x) PARAMS (x) SEP()
  POST [ int_or_ptr_type ]  PROP() RETURN (x) SEP().
\end{lstlisting}

Each of the functions in \autoref{int-or-ptr} has a simple implementation.
Their specifications in \autoref{specs-int-or-ptr}, in the VST format,
are straightforward.  For example, \lstinline{test_int_or_ptr_spec}
says that the \lstinline{test_int_or_ptr} function has no
spatial footprint (\lstinline{SEP()}), that the C function's sole parameter
has the C type which VST calls \lstinline{int_or_ptr_type}, and that the
CompCert value \lstinline{x} passed in that parameter
is a \lstinline{valid_int_or_ptr}---meaning that if it's a pointer, then its low-order bit
is 0, and if it's an integer, then 1.

Then the specification of \lstinline{int_or_ptr_to_int} has a precondition
that its input value $x$ is an integer value, not a pointer value.  
To satisfy that precondition, the C program could call
\lstinline{test_int_or_ptr} before calling either \lstinline{int_or_ptr_to_int}
or \lstinline{int_or_ptr_to_ptr}.

In the standard C semantics (provided by VST), it would be impossible to prove
that \lstinline{test_int_or_ptr} satisfies its specification.  But its
specs are \emph{consistent} with the C standard, i.e., they are a refinement
of the standard.  So we can safely axiomatize it, as a way of claiming
that the C compilers we care about (gcc, clang, CompCert) guarantee this extension
to the standard.

The C specification permits casting an integer to a pointer, provided that no
pointer operations are performed on it before it is cast back to an integer---and vice versa.
Thus, the C semantics does permit holding an integer in a variable of type \lstinline{void*}.
So the remaining functions can be proved correct using the C semantics,
though at present we have axiomatized them as well.

VST's separation logic assumes that a value of type \lstinline{void*} is
a pointer value, not an integer value.  To indicate to VST that our \lstinline{value}
type can hold either a pointer or an integer
(which VST calls \lstinline{int_or_ptr_type})
we mark the \lstinline{void*} type with a special attribute
(see \autoref{listing-value}).  The attribute we use
is does not alter the C semantics, but serves as a hint to VST about
how to treat variables and fields of this type.

\section{The C semantics of range testing}\label{sec:rangetest}
The \lstinline{forward} function (\autoref{list:forward}, line 6)
tests whether a pointer \lstinline{v} points within the from-space,
using the \lstinline{Is_from} function defined as follows:
\begin{lstlisting}[language=C]
int Is_from(value* from_start, value * from_limit,  value * v)
  { return (from_start <= v && v < from_limit); }
\end{lstlisting}
Pointer range testing is a standard operation in many garbage collectors,
but it is illegal in the standard C semantics!
The pointer comparisons $p <q$ or $p\le q$ in the official semantics of C require that $a$ and $b$ both be pointers into (or at the end of)
the same allocated block, otherwise the comparison is \emph{undefined
behavior}.  C requires this to avoid situations such as,
\begin{lstlisting}
    p = (char *)malloc(n);
    if (...) free(p);
    q = (char *)malloc(n);
    if (p<=q) {...}
\end{lstlisting}  
If \lstinline{free(p)} is executed, this will have different behavior depending on whether \lstinline{malloc}
re-uses the same block.  If \lstinline{free(p)} is not executed,
this tests whether \lstinline{malloc} chooses to place the new
block before or after the old block.  The designers of the C
semantics wanted to prevent (legal) programs from depending on
such implementation choices in \lstinline{malloc}.  Similar considerations
apply when \lstinline{p} or \lstinline{q} are pointers to 
stack-allocated local variables.

However, to test whether \lstinline{p} points within the from-space,
we need the range comparison \lstinline{from_start$\le$p<from_limit},
where it is guaranteed that \lstinline{from_start} and \lstinline{from_limit}
are in the same malloc'ed block and that \lstinline{p} points within
some malloc'ed block (i.e., generation).  This comparison is illegal
in the standard C semantics, but essentially all garbage collectors
must do it and it is assured to work properly in all C compilers---and it cannot inappropriately test the behavior of \lstinline{malloc}.
Therefore it is safe to extend the semantics of C to axiomatize
the behavior of this range comparison.

For the comparison \lstinline{$p$ <= $q$ && $q$ < $r$},
our axiom has the precondition 
that $p$ and $r$ are in (or point just past) the same \lstinline{malloc}'ed block,
and that $q$ points within (or just past) some allocated block.
We can express these conditions as predicates in the Verifiable C
separation logic.

\end{document}